\documentclass[acmsmall,screen,nonacm]{acmart}
\usepackage{booktabs}
\usepackage{graphicx}
\usepackage[ruled,vlined,linesnumbered]{algorithm2e}
\usepackage{amsmath}
\usepackage{multirow}
\usepackage{makecell} 
\usepackage{subcaption}
\usepackage{pgfplots}
\pgfplotsset{compat=1.18}
\usetikzlibrary{arrows.meta}
\usepackage{mathpartir}
\usepackage[capitalize]{cleveref}
\usepackage[dvipsnames]{xcolor}
\usepackage{stfloats}
\usepackage{float}

\usepackage{stmaryrd}

\usepackage{mathtools}

\usepackage{tikz}
\usetikzlibrary{quantikz2}

\renewcommand{\>}{\rangle}

\newcommand{\bona}{\textsc{Bona}}
\newcommand{\recycle}{\textsc{Recycle}}
\newcommand{\quantinuum}{\textsc{Decross}}
\newcommand{\trivD}{\textsc{TrivD}}

\newcommand{\borr}[2]{\mathbf{borrow}[{#1}] \{{#2}\}}
\newcommand{\qv}{\mathit{qv}}
\newcommand{\qa}{\mathit{qa}}
\newcommand{\argsop}{\mathit{args}}
\newcommand{\comm}[2]{{#1}\, \|\, {#2}}

\newcommand{\RuleAssoc}{\textsc{Assoc}}
\newcommand{\RuleCompLeft}{\textsc{CompL}}
\newcommand{\RuleCompRight}{\textsc{CompR}}
\newcommand{\RuleBorrInner}{\textsc{BorrI}}
\newcommand{\RuleBorrLeft}{\textsc{BorrL}}
\newcommand{\RuleBorrRight}{\textsc{BorrR}}
\newcommand{\RuleBorrComm}{\textsc{BorrC}}
\newcommand{\RuleSwap}{\textsc{Swap}}
\newcommand{\RuleSubst}{{\textsc{Borr}}}
\newcommand{\TopoRules}{\textsc{TopoRules}}
\newcommand{\topotrans}{\mathrel{\to_{\mathsf{T}}}}

\newcommand{\In}{\mathsf{In}}
\newcommand{\Out}{\mathsf{Out}}
\newcommand{\Op}{\mathsf{Op}}
\newcommand{\fst}{\mathsf{fst}}
\newcommand{\lst}{\mathsf{lst}}

\newcommand{\dep}{\mathrm{dep}}
\newcommand{\depthdelta}{\Delta\mathrm{d}}
\newcommand{\estdepthdelta}{\widehat{\depthdelta}}
\newcommand{\Pool}{\mathit{Pool}}
\newcommand{\Skipped}{\mathit{Skipped}}

\newcommand{\sem}[1]{\left \llbracket #1 \right \rrbracket}

\newcommand{\T}{\mathcal{T}}

\newcommand{\tX}{\mathtt{X}}

\newcommand{\tT}{\mathtt{T}}

\newcommand{\tCNOT}{\mathtt{CNOT}}

\newcommand{\tMCX}[1]{\mathtt{C}^{#1}\mathtt{NOT}}

\newcommand{\xf}[0]{x_{\text{false}}}
\newcommand{\fr}[0]{\textsc{Complete-Borrowing}}
\newcommand{\daginudg}[0]{\textsc{DAG-in-UDG}}

\newcommand{\blue}[1]{{\color{blue}#1}}

\title{\bona: Automatic Management of Dirty Ancilla Borrowing in Quantum Circuits}

\author{Xiaoquan Xu}
\authornote{These authors contributed equally.}
  \orcid{0009-0008-0285-695X}
	\affiliation{
    \institution{Institute of Software, Chinese Academy of Sciences and University of Chinese Academy of Sciences
		}
		\country{China}
	}
	\email{xuxq@ios.ac.cn}

\author{Chenke Liu}
\authornotemark[1]
    \orcid{0009-0009-5703-6614}
	\affiliation{
    \institution{Institute of Software, Chinese Academy of Sciences and University of Chinese Academy of Sciences
		}
		\country{China}
	}
	\email{liuchenke24@mails.ucas.ac.cn}

\author{Boning Meng}
\authornotemark[1]
    \orcid{0009-0006-0088-1639}
	\affiliation{
    \institution{Institute of Software, Chinese Academy of Sciences and University of Chinese Academy of Sciences
		}
		\country{China}
	}
	\email{mengbn@ios.ac.cn}

\author{Zihao Shen}
\authornotemark[1]
    \orcid{0009-0004-9115-6250}
	\affiliation{
    \institution{Institute of Software, Chinese Academy of Sciences and University of Chinese Academy of Sciences
		}
		\country{China}
	}
	\email{shenzihao25@mails.ucas.ac.cn}
    
\author{Li Zhou}
\authornote{Corresponding author: Li Zhou.}
  \orcid{0000-0002-9868-8477}
\affiliation{
    \department{Key Laboratory of System Software (Chinese Academy of Sciences) and State Key Laboratory of Computer Science}
    \institution{Institute of Software, Chinese Academy of Sciences}
    \country{China}
}
\email{zhouli@iscas.ac.cn}
\email{zhou31416@gmail.com}

\begin{document}

\begin{abstract}

The management of ancilla qubits has become a critical technique for reducing quantum circuit width.
Dirty ancillas, which may be borrowed from any temporarily idle qubit regardless of their initial states, offer substantial flexibility for width optimization, but their use has so far required manual and error-prone handling.
We formalize the dirty-qubit borrowing problem and establish a fundamental computational limit by proving its NP-hardness.
To support practical optimization, we present \bona, the first scheduler for dirty-qubit borrowing, built on a novel depth-aware heuristic algorithm.
We evaluate \bona~ across a variety of benchmarks, including practical quantum circuits and randomly arranged compositions of real circuit modules, and find that it reduces nearly 99\% of dirty ancillas on average with controlled depth overhead.
In particular, for parallel quantum walk---an essential component of parallel Hamiltonian simulation---\bona~ matches the circuit width achieved by the clean-qubit schemes of
\citeauthor{jiang2024recycling}~(\citeyear{jiang2024recycling}) and
\citeauthor{quantinuum}~(\citeyear{quantinuum}),
but attains significantly smaller circuit depth, 
providing concrete evidence that dirty ancillas offer unique optimization advantages in circuits with certain parallelism.

\end{abstract}

\maketitle

\section{Introduction}

As quantum computing scales toward practical applications~\cite{Preskill2025BeyondNISQ}, the management of \emph{temporary qubits}---known as \emph{ancillas}---has become a central concern in circuit compilation and optimization~\cite{Paler2017Wirerecycling, quantinuum, Hua2023CqQR, jiang2024recycling}. Ancillas are essential for efficiently realizing complex quantum circuits with low \emph{depth} (runtime after parallelization) and moderate \emph{width} (qubit count) through a two-step process~\cite{Barenco1995Elementarygates, Paler2017Wirerecycling}: in the \emph{basic construction} step, ancillas enable multi-qubit gate decomposition and the implementation of functional modules with greatly reduced depth and size, trading only a reasonable, often constant or constant-factor, increase in width for polynomial or even exponential depth reduction~\cite{nie2024quantumcircuitmultiqubittoffoli, baker2019decomposingquantumgeneralizedtoffoli, factoring2np2_2017, gidney2018factoringn2cleanqubits}; in the \emph{qubit management} step, ancilla recycling and reuse in larger modular circuits can amortize this width overhead, reducing the number of required ancillas---sometimes exponentially or even to a constant---while adding only minor or constant-factor depth growth~\cite{Paler2017Wirerecycling, quantinuum, Hua2023CqQR, jiang2024recycling}.

Ancilla qubits fall into two categories: clean ancillas, initialized to the ground state $|0\>$, and dirty ancillas, whose initial states are unknown but must be fully restored after use~\cite{Barenco1995Elementarygates}.
Dirty ancillas provide much greater flexibility in qubit management. Since they do not require a $|0\rangle$ initialization, they can borrow any qubit idle at that time to accomplish their gates. When this happens, as in~\Cref{fig:intro-reuse-eg} , a qubit $d$ is saved and width is reduced by one
  (\cref{fig:intro-reuse-eg}).
\noindent Consequently, dirty ancillas are a recognized and indispensable technique in circuit design, including elementary gates~\cite{zindorf2025efficientimplementationmulticontrolledquantum, nie2024quantumcircuitmultiqubittoffoli}, unitary synthesis~\cite{Low2024tradingtgatesdirty}, arithmetic circuits~\cite{gidney2018factoringn2cleanqubits, factoring2np2_2017}, and cryptography~\cite{Ha2024Resourceanalysis}. Recent work~\cite{Low2024tradingtgatesdirty, Huang2025Constructingquantumimplementations} further shows that dirty ancillas support unique reduction in T-gate count and depth. These features make dirty ancillas a key mechanism for depth-width optimization in practical quantum circuits.

\begin{figure}[tb]
    \centering
    \vspace{-0.2cm}
    \includegraphics[width=.7\linewidth]{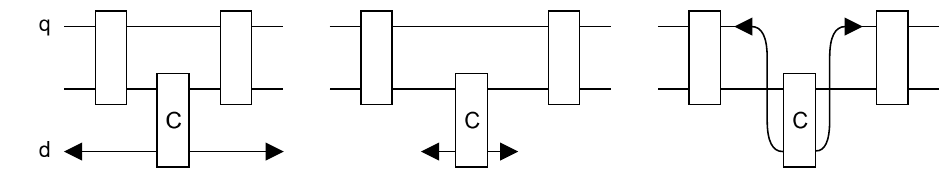}
    \vspace{-0.4cm}
    \caption{The transformation from left to right illustrates how an idle middle interval of the top wire can be borrowed to serve as a dirty ancilla (the wire marked with $\blacktriangleleft$ and $\blacktriangleright$).}
    \label{fig:intro-reuse-eg}
\end{figure}

While dirty ancillas have been well studied in the basic construction step, research on their qubit management remains incomplete.
% Despite their importance, research on ancilla management remains incomplete. 
Unlike clean ancillas, whose recycling has been extensively studied~\cite{Paler2017Wirerecycling, quantinuum, Hua2023CqQR, jiang2024recycling}, dirty-ancilla management is typically handled manually by algorithm developers
% currently rely on  to manually exploit their flexibility 
in specific arithmetic or modular constructions~\cite{factoring2np2_2017, Huang2025Constructingquantumimplementations}.
%are understood mostly from a circuit-design perspective, 
% where algorithm developers manually exploit their flexibility in specific arithmetic or modular constructions~\cite{factoring2np2_2017, Huang2025Constructingquantumimplementations}.
Systematic study of \emph{dirty-qubit borrowing}---the automatic scheduling of dirty ancilla reuse during compilation to further reduce circuit width---remains largely absent.
Existing compilers typically assume all ancillas are clean~\cite{quantinuum, Hua2023CqQR, fang2023Dynamic, qiskit}, while manual dirty management is error-prone and infeasible for large-scale circuits.
Q\#~\cite{qpunch} already allows programmers to explicitly declare the use of dirty ancillas, but provides little mechanism for managing them.
This lack of compiler-level support prevents the full realization of dirty ancillas' potential for resource optimization.

To bridge this gap, it is essential to address two fundamental problems that jointly capture the formal and practical aspects of dirty-qubit borrowing:
\begin{enumerate}
  \item \emph{Formalization and theoretical limits}: How can the dirty-qubit borrowing problem be rigorously defined, and what are its limits in terms of computational complexity?
  \item \emph{Algorithm design}: How to develop an efficient and effective algorithm for dirty-qubit borrowing, and under what conditions do dirty ancillas clearly outperform clean ones?
\end{enumerate}

% \ckl{need to modify}
\paragraph{{\textbf{This work}}}
This work aims to establish a complete foundation for automatic dirty-ancilla management in quantum circuit compilation. 
We first formalize the dirty-qubit borrowing problem and prove that its decision version is NP-complete.
Motivated by this hardness, we present \bona, the first compiler-level scheduler for dirty-qubit borrowing, and use several benchmarks to evaluate its effectiveness and identify when dirty ancillas offer optimization advantages.

% which provides evidence that dirty ancillas can outperform clean ones in both circuit depth and width on practical benchmarks.

\paragraph{Formalization of dirty-qubit borrowing}
In a structured circuit, a dirty ancilla $d$ specifies a code block $C$ when it is declared.
Since $C$ guarantees to restore $d$'s initial state, any
temporarily idle qubit $q$ can safely act as $d$. We formalize this
\emph{borrowing transition} as substituting all occurrences of $d$ with $q$
(as shown in Fig.~\ref{fig:intro-reuse-eg}):
$$\borr{d}{C} \to C[q/d], \quad\mbox{if }q\mbox{ is not involved in }C.$$
When the dirty ancilla $d$ in a circuit $C$ is correctly used (in design)~\cite{su2024bibasedreasoningquantumprograms, su2025borrowingdirtyqubitsquantum}, performing a borrowing transition yields a semantically equivalent circuit.
In addition to the borrowing transition, we further introduce several semantics-preserving \emph{structural transitions} and \emph{topological transitions} to describe how a circuit may change.
Based on these transitions, and from the perspective of circuit-width optimization, we define the \emph{dirty-qubit borrowing problem} as follows: given a circuit, find a sequence of transitions to minimize the circuit width.

To rigorously reason about this search space, we abstract the circuit into an
\emph{Endpoint DAG model} that is invariant under structural and topological
transitions and reformulates borrowing as a graphical \emph{edge-splicing}
operation. We prove a sound-and-complete step correspondence between edge
splicing and circuit-level borrowing up to topological preprocessing, thereby
establishing a unified operational foundation for our subsequent theoretical
analysis and algorithm design.

\paragraph{NP-hardness of the borrowing problem}
We first identify a tractable fragment: if topological transitions are
disallowed, the minimal circuit width can be exactly computed in linear time
via a static structural recursion. However, once topological transitions are
permitted, the optimal topological structure must be dynamically generated
during the borrowing process. This interleaving of topological rearrangements
and borrowing choices creates a combinatorial explosion.

We formally establish the NP-completeness of the dirty-qubit borrowing problem
by proving that finding a complete borrowing assignment on the Endpoint DAG is
equivalent to finding an acyclic realization of an intermediate Uncertain
Directed Graph (UDG), which we show is NP-hard via a reduction from 3-SAT.
Furthermore, we show that attempting to optimize circuit depth--whether
as the primary or secondary objective--also runs into an NP-hard combinatorial
wall. This motivates the need for a pragmatic, heuristic-driven scheduler.

\paragraph{Heuristic borrowing scheduling}
We present \bona, the first automated tool for dirty-qubit
borrowing scheduling, which achieves near-optimal width reduction in
practice through a depth-aware heuristic combined with a width-greedy
strategy.

The scheduler combines three key ideas:
(i) a depth-aware heuristic that instantly estimates the impact of each borrowing step and prioritizes depth-preserving embeddings;
(ii) a dynamic resource pool that efficiently processes queries and updates;
and (iii) a \emph{Clean-then-Dirty} orchestration strategy that allows \bona\ to directly couple with existing state-of-the-art clean-ancilla optimizers,
functioning as a highly effective end-to-end compiler pipeline.
We also prove that, for a safe input circuit, every successful splice performed
by \bona\ satisfies the unreachable condition, ensuring that its output
represents a circuit semantically equivalent to the input.

\paragraph{Implementation and case studies}
We implement \bona\ as a Python tool based on the heuristic approach.
We evaluate \bona~ across three levels of benchmarks: 
(1) two quantum algorithms, parallel quantum walk (PQW)~\cite{Zhang2024parallelquantum} and Shor's algorithm~\cite{Shor_1997} using the implementation from~\cite{factoring2np2_2017};
(2) component-level quantum circuits from~\cite{gidney2015, Low2024tradingtgatesdirty}; and 
(3) randomly stitched parallel circuits constructed from RevLib~\cite{revlib}.
These benchmark circuits are represented at the logical-gate level and their scale is relevant to near-term hardware demonstrations.
When a benchmark circuit contains both clean and dirty ancillas, we first apply \recycle~\cite{jiang2024recycling} to reuse the clean ancillas and then apply \bona\ to reuse the dirty ancillas.
Our key experimental findings fall into two aspects:

\begin{enumerate}
    \item \textit{Effectiveness.}
    On the PQW circuits, \recycle+\bona\ reduces circuit width by 92\%--99\% and dirty-ancilla usage by 99\%--100\%, with a depth overhead of 101\%--330\%.
    When applied after manual preprocessing, it eliminates all dirty ancillas while limiting the depth overhead to 1\%--18\%.
    On Shor's algorithm, \bona\ reduces the dirty-ancilla count from 192--6400 to only 3--20, with a depth overhead of at most 20\%, producing circuits close to the manually optimized implementations.
    For the component-level benchmarks, \recycle+\bona\ matches the final width of manual optimization in three out of four cases, with depth overheads between 0\% and 27\%.

    \item \textit{Scheduling Advantage of Dirty Ancillas.}
    For PQW circuits, we find that the dirty-ancilla implementations optimized by \recycle+\bona\ achieve width comparable to that of the clean-ancilla implementations optimized by \recycle~\cite{jiang2024recycling} or \quantinuum~\cite{quantinuum}, while attaining smaller final depth at most scales, despite having greater depth before optimization.
    In the randomly composed circuits, this advantage emerges beyond a certain degree of parallelism.
    At the highest tested parallelism, the dirty-ancilla circuits optimized by \recycle+\bona\ achieve nearly half the depth of their clean-ancilla counterparts.

\end{enumerate}

Our logical-gate-level evaluation demonstrates that \bona\ effectively reuses dirty ancillas and can be combined with clean-ancilla optimizers such as \recycle. In practice, clean and dirty ancillas can play complementary roles, with dirty ancillas offering additional scheduling flexibility that may reduce post-optimization depth.

\paragraph{{\textbf{Organization and summary of contributions}}}
\Cref{sec: preliminary} introduces the basic concepts of quantum circuits and dirty ancilla qubits.
The subsequent sections present our main contributions:
\begin{itemize}
    \item \Cref{sec:formalization} formally defines the dirty-qubit borrowing problem on a circuit language;
    \item \Cref{sec:endpointdag} introduces the Endpoint DAG model, translating syntactic borrowing into a graphical edge-splicing operation;
    \item \Cref{sec:hardness} proves the NP-hardness of dirty-qubit borrowing problem;
    \item \Cref{sec:heuristic} introduces \bona, a novel depth-aware heuristic algorithm;
    \item \Cref{sec:evaluation} presents case studies.
\end{itemize}
We finally discuss related work in \Cref{sec: related work}.

\section{Preliminary}
\label{sec: preliminary}

\subsection{Quantum Circuits}
\paragraph{Quantum gates.}
The state of an $n$-qubit system is described by a vector in Hilbert space $(\mathbb{C}^2)^{\otimes n}$~\cite{Nielsen_Chuang_2010}.
A \emph{quantum gate} is a fixed-arity unitary operator. A $k$-qubit gate is represented by $2^k$-dimensional matrix, 
such as single-qubit NOT gate $\tX$ and multi-qubit controlled gates $\tMCX{k}$:
\[
\tX = \begin{pmatrix} 0 & 1 \\ 1 & 0 \end{pmatrix}, \quad
\tMCX{1} = \tCNOT = \left(\begin{smallmatrix} 1 & 0 & 0 & 0 \\ 0 & 1 & 0 & 0 \\ 0 & 0 & 0 & 1 \\ 0 & 0 & 1 & 0 \end{smallmatrix}\right), \quad
\tMCX{k+1} =  \begin{pmatrix} 1 & 0 \\ 0 & 0 \end{pmatrix} \otimes I_k + \begin{pmatrix} 0 & 0 \\ 0 & 1 \end{pmatrix} \otimes \tMCX{k}
\]
where $I_k$ is the $2^k$-dimensional identity matrix.
Given a sequence of distinct qubits $\overline{q}=(q_1,\dots,q_k)$, we write $U[\overline{q}]$ to denote the application of a $k$-qubit gate $U$ on those qubits, leaving all others unchanged.

\paragraph{Quantum circuits.}
A \emph{quantum circuit} $C$ is a finite sequence of gate applications that jointly form a unitary operator acting on the qubits appearing in~$C$, 
commonly denoted by an intuitive graphical representation. For example, the left panel in~\Cref{fig:circ-dag-example} represents a circuit that first applies two NOT gates $\tX[q_1]$ and $\tX[q_2]$, and then follows a controlled-NOT gate $\tMCX{1}[q_1,q_2]$,
where the control qubit is marked by a solid bullet and the target qubit by an $\oplus$ symbol.
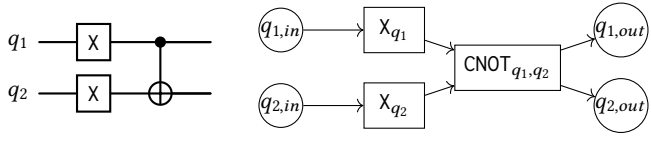
\begin{figure}
\centering
\begin{minipage}{0.2\textwidth}
\centering
\small
\begin{quantikz}[row sep = 0.2cm]
   \lstick{$q_1$} \qw & \gate{\tX} & \ctrl{1} & \qw \\
   \lstick{$q_2$} \qw & \gate{\tX} & \targ{}  & \qw
\end{quantikz}
\end{minipage}
\begin{minipage}{0.45\textwidth}
\centering
\begin{tikzpicture}[
node distance=1.5cm,
every node/.style={font=\small},
gate/.style={rectangle, draw=black, minimum width=0.8cm, minimum height=0.6cm},
io/.style={circle, draw=black, minimum size=0.55cm, inner sep=0pt}
]

\node[io] (q1_in) at (0, 1) {$q_{1,in}$};
\node[io] (q2_in) at (0, 0) {$q_{2,in}$};

\node[gate] (x1) at (1.5, 1) {$\tX_{q_1}$};
\node[gate] (x2) at (1.5, 0) {$\tX_{q_2}$};

\node[gate] (cnot) at (3, 0.5) {$\tCNOT_{q_1,q_2}$};

\node[io] (q1_out) at (4.5, 1) {$q_{1,out}$};
\node[io] (q2_out) at (4.5, 0) {$q_{2,out}$};

\draw[->] (q1_in) -- (x1);
\draw[->] (x1) -- (cnot);
\draw[->] (cnot) -- (q1_out);

\draw[->] (q2_in) -- (x2);
\draw[->] (x2) -- (cnot);
\draw[->] (cnot) -- (q2_out);

\end{tikzpicture}
\end{minipage}
\caption{A quantum circuit (left) with input qubits $q_1$ and $q_2$ and three quantum gates, and its corresponding DAG representation (right).}
% \vspace{-0.1cm}
\label{fig:circ-dag-example}
\end{figure}

\begin{definition}[Basic notions in Graph Theory]
A \emph{directed graph} is a pair $G=(V,E)$, where $V$ is a finite set of
vertices and $E$ is a (possibly multi-)set of directed edges. An edge $e =
(u,v) \in E$ directs from a \emph{predecessor} $u$ to a \emph{successor} $v$.
For any vertex $v\in V$, its \emph{in-degree} $\deg^{-}(v)$ and
\emph{out-degree} $\deg^{+}(v)$ denote the number of incoming and outgoing
edges (counted with multiplicity), respectively.

A \emph{path} $p$ in $G$ is a sequence of vertices
$(v_1,\dots,v_\ell)$ such that $(v_i,v_{i+1})\in E$ for each $i$.
We write $a\leadsto b$ if $b$ is \emph{reachable} from $a$, i.e.,
there exists a path that starts from $a$ and ends at $b$.

A \emph{cycle} is a path $(v_1,\dots,v_\ell)$ with $\ell\ge2$ and $v_1=v_\ell$.
A graph $G$ is a \emph{directed acyclic graph (DAG)} if it contains no cycles.

\end{definition}

\paragraph{Directed acyclic graph (DAG) representation.}
Abstracting circuits into directed acyclic graphs to analyze temporal
dependencies is rooted in classical static timing
analysis~\cite{sapatnekar2004timing} and is widely adopted by modern quantum
compilation frameworks (e.g., Qiskit~\cite{qiskit}). In a general DAG
representation of a quantum circuit $C$ (e.g., the right panel
in~\cref{fig:circ-dag-example}), vertices may correspond to quantum
gates (e.g., the boxed $\tX_{q_1}$, $\tX_{q_2}$, and $\tCNOT_{q_1,q_2}$) and input/output nodes
(e.g., the four circled nodes).
Directed edges represent temporal dependencies induced by consecutive
operations along each qubit wire, including those from an input to the first
gate on a wire (e.g., $q_{1,in} \to \tX_{q_1}$), those between two gates
(e.g., $\tX_{q_1} \to \tCNOT_{q_1,q_2}$), and those from the last gate on a
wire to the corresponding output (e.g., $\tCNOT_{q_1,q_2} \to q_{1,out}$).
A bare dependency graph does not recover wire/port information; our Endpoint
DAG and its validity conditions are defined in Section~\ref{sec:endpointdag}.

\paragraph{Circuit width and depth.}
We use width and depth as coarse-grained logical-level resource metrics:
\begin{itemize}
  \item \textit{Width}: the number of qubits used by the circuit, including both data and ancillas.
  \item \textit{Depth}: the length of the longest path in the circuit's DAG.
\end{itemize}
Width corresponds to the spatial cost---how many qubits are required simultaneously---whereas depth reflects the temporal cost.
Reducing one often increases the other, making the width--depth trade-off a central challenge in quantum compilation and the context in which ancilla management plays a key role.

\subsection{Dirty Ancilla Qubits}

\begin{figure}[t]
  \centering
  \includegraphics[width=0.95\linewidth]{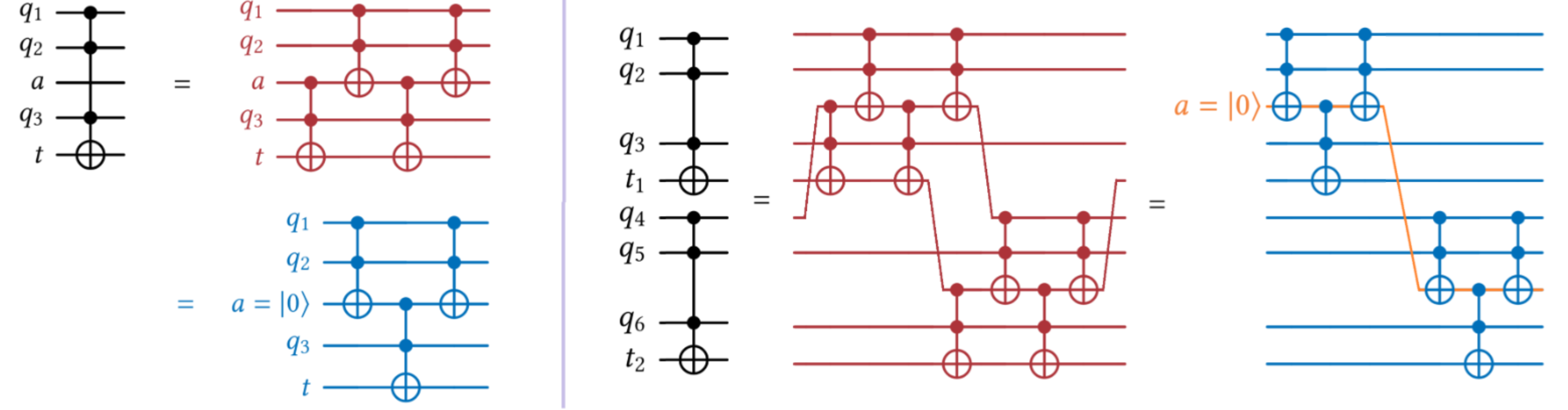}
  \caption{Left: Implementations of $\tMCX{3}$ using $\tMCX{2}$ gates and one ancilla. The red dirty-ancilla circuit is fully semantically equivalent to the black target circuit, whereas the blue clean-ancilla circuit is equivalent only when the ancilla is initialized to $\ket{0}$.
  Right: Implementation of two parallel $\tMCX{3}$ with recycling, with the clean ancilla marked by an orange line.
  }
  \label{fig:mcx pre}
\end{figure}

A \emph{dirty ancilla} is an auxiliary qubit whose initial state is arbitrary and unknown. 
To ensure safe use, a dirty ancilla must be fully restored to its original (unknown) state, 
so as to avoid irreversible corruption of the borrowed qubit during the management step.
Formally, given a circuit that syntactically uses a dirty ancilla~$a$ and represents a unitary operator~$U$,
we say that $U$ \emph{safely uses the dirty ancilla}~$a$ if
\begin{equation}
\label{eq:safe-dirty}
U = I_{a} \otimes V
\quad\text{for some unitary } V.
\end{equation}

Several equivalent formulations of this safety condition have been discussed in prior work~\cite{su2025borrowingdirtyqubitsquantum}.
The left panel of~\Cref{fig:mcx pre} illustrates the safe use of a dirty ancilla~$a$ (in red) and a clean ancilla~$a$ (in blue) 
in the implementation of a $\tMCX{3}$ gate (leftmost circuit) on $(q_1, q_2, q_3, t)$ using only $\tMCX{2}$ gates (a.k.a.\ Toffoli gates).

\paragraph{Dirty-qubit borrowing.}
Once a dirty ancilla is safely used, the condition in~\Cref{eq:safe-dirty} ensures that the circuit acts trivially on the borrowed qubit.
Therefore, any temporarily idle qubit can be \emph{borrowed} as a dirty ancilla. 
This flexibility enables better width optimization and is often employed when the circuit width is fixed~\cite{Low2024tradingtgatesdirty, factoring2np2_2017} or preferentially constrained.
The right panel of~\Cref{fig:mcx pre} illustrates the implementation of two parallel $\tMCX{3}$ gates.
The dirty scheme (shown in red) yields a depth-7 implementation without allocating any additional ancillas, since the two $\tMCX{3}$ gates can borrow each other's working qubits ($q_4$ and $t_1$ here) as temporary ancillas.
In contrast, clean ancillas cannot be borrowed from working qubits, resulting in a sequential implementation that reuses the same clean ancilla~$a$ to minimize width.

\paragraph{Effectiveness: dirty vs. clean ancillas}
For basic circuit building blocks such as multi-controlled NOT gates, carry modules, and constant adders, clean ancillas generally achieve smaller circuit size and lower depth.
Dirty ancillas usually require a comparable number of ancillas but incur a constant-factor overhead in depth, often around twofold, due to the additional unitaries needed to restore arbitrary initial states rather than the fixed~$\ket{0}$ state.

On the other hand, dirty ancillas offer much greater flexibility in resource management, 
which can be advantageous for large, sparse, or parallel circuits. 
For example, to implement $2n$ parallel $\tMCX{3}$ gates, the dirty scheme achieves a depth-7 implementation without allocating any ancillas by pairing every two $\tMCX{3}$ gates as shown in~\Cref{fig:mcx pre}. 
In contrast, the clean scheme requires at least one additional ancilla, 
and its depth is approximately~$6n/k$ when at most~$k$ clean ancillas are available, 
which remains worse than the dirty scheme in both depth and width.

In summary, efficient and effective management of dirty ancillas is crucial, and it remains open whether more realistic scenarios exist where dirty ancillas clearly outperform clean ones.

\section{Formal Model of Dirty-Qubit Borrowing}\label{sec:formalization}

This section formalizes dirty-qubit borrowing in quantum circuits. We first define a small language with explicitly scoped dirty ancillas, and then formulate borrowing as a transition-based finding problem.

\subsection{Syntax and Semantics}

We begin by introducing a simple language that supports the declaration and use of dirty ancillas under explicit scoping.

\begin{definition}[Syntax] A quantum circuit with dirty ancillas is generated by:
\begin{align*}
  (\text{Circuit}) \quad 
  C ::= &\ U[\overline{q}] 
        \mid C_1; C_2 
        \mid \borr{d}{C}
\end{align*}
\end{definition}

A quantum circuit is either a gate $U$ acting on a sequence of qubits $\overline{q}$, a sequential composition of two circuits, or a statement $\borr{d}{C}$, which borrows a qubit and uses it as a dirty ancilla $d$ in $C$.
The set of quantum variables $\qv(C)$ of a circuit $C$ is inductively defined as follows:
\begin{align*}
\qv(C) = \begin{cases}
      \{q \mid q \in \overline{q}\}, & \text{if } C = U[\overline{q}]\\
      \qv(C_1) \cup \qv(C_2), & \text{if } C = C_1; C_2\\
      \qv(C') \setminus \{d\}, & \text{if } C = \borr{d}{C'}
    \end{cases}
\end{align*}

Since a dirty ancilla may be obtained by borrowing a temporarily idle qubit that may carry useful data, the circuit $C$ must use it \emph{safely}. That is, the borrowed qubit must be restored to its original state when it goes out of scope, even though that state is unknown to $C$. 
Formally, safety here is characterized by \Cref{eq:safe-dirty}. We assign the distinguished value $\bot$ to any unsafe use of dirty ancillas, and any circuit containing such an invalid part is itself deemed invalid.

\begin{definition}[Semantics]\label{def:semantics}
    The semantics of a quantum circuit $C$, denoted by $\sem{C}$, is defined as a unitary operator on $\qv(C)$ when $C$ is \textit{safe}, and as $\bot$ when $C$ is \textit{unsafe}. Moreover, $\bot U = U \bot = \bot$ and $\bot \otimes U = U \otimes \bot = \bot$ for any unitary operator $U$. Formally:

    \begin{enumerate}
        \item $\sem{U[\overline{q}]} = U_{\overline{q}}$;
        
        \item $\sem{C_1; C_2} = 
        \bigl( \sem{C_2} \otimes I_{\qv(C_1) \setminus \qv(C_2)} \bigr)
        \bigl( \sem{C_1} \otimes I_{\qv(C_2) \setminus \qv(C_1)} \bigr)$;
        
        \item $\sem{\borr{d}{C}} = 
            \begin{cases}
                F_{\qv(C) \setminus \{d\}}, & \text{if } \sem{C} = F_{\qv(C) \setminus \{d\}} \otimes I_d, \\
                \sem{C}, & \text{if } d \notin \qv(C), \\
                \bot, & \text{otherwise.}
            \end{cases}$
    \end{enumerate}
\end{definition}

\paragraph{Substitution of qubits.}
The substitution $C[q/d]$ is defined inductively by replacing each free occurrence of qubit $d$ in circuit $C$ with $q$, while leaving bound occurrences of $d$ unaffected:
\begin{align*}
    (U[\overline{p}])[q/d] 
      &\triangleq U[\overline{p}[q/d]] \\[4pt]
    (\borr{a}{C})[q/d] 
      &\triangleq 
        \begin{cases}
            \borr{a}{C}, & \text{if } d = a; \\[3pt]
            \borr{a'}{C[a'/a][q/d]}, & \text{if } q = a \text{ and }a'\text{ is a fresh name}\\[3pt]
            \borr{a}{C[q/d]}, & \text{otherwise.}
        \end{cases} \\[6pt]
    (C_1; C_2)[q/d] 
      &\triangleq C_1[q/d];\, C_2[q/d].
\end{align*}
The usual substitution lemma holds: $\sem{C[q/d]} = \sem{C}[q/d]$.

For simplicity, from now on we assume that every dirty ancilla in a circuit
has an identifier distinct from those of all other dirty ancillas and working
qubits, and is actually used within its borrowing body
(i.e., $d\in\qv(C)$ for every $\borr{d}{C}$).
This assumption can always be ensured by renaming each dirty ancilla with a fresh name
and removing any unused borrowing statements, both of which preserve the semantics of any circuit.
For clarity, we write $\qa(C)$ to denote the set of all dirty ancillas in circuit $C$, and $\argsop(C) \triangleq \qv(C) \cup \qa(C)$.

\subsection{Circuit transitions and the dirty-qubit borrowing problem}\label{subsec:transitions}

\paragraph{Objective of dirty-qubit borrowing}
Given a safe circuit $C$, \textit{dirty-qubit borrowing} aims to find a circuit $C'$ that is semantically equivalent to $C$ while containing the fewest $\mathbf{borrow}$ statements, i.e., the fewest required dirty ancillas.
We model such a finding process as the repeated application of semantics-preserving transition rules.

\paragraph{Structural transitions}
Structural transitions lift transitions on subcircuits to transitions on the entire circuit.
\begingroup
\setlength{\jot}{0pt}
\begin{mathpar}
    \inferrule*[right=\RuleCompLeft]{
        C_1 \to C_1'
    }{
        C_1; C_2 \to C_1'; C_2
    }
    \hspace{1.5em}
    \inferrule*[right=\RuleCompRight]{
        C_2 \to C_2'
    }{
        C_1; C_2 \to C_1; C_2'
    }
    \hspace{1.5em}
    \inferrule*[right=\RuleBorrInner]{
        C \to C'
    }{
        \borr{d}{C} \to \borr{d}{C'}
    }
\end{mathpar}

\endgroup

Since structural transitions only propagate local transitions to larger circuits, we treat them as implicit and do not mention them further.
For example, when we say that $C \to C'$ by one application of a non-structural transition rule, we mean that $C'$ is obtained by applying that rule once together with any number of structural transitions.

\paragraph{Topological transitions}
The following rules constitute the set $\TopoRules$ and capture semantics-preserving structural rearrangements of circuits:
\begingroup
\setlength{\jot}{5pt}
\begin{mathpar}
    \inferrule*[right=\RuleBorrLeft]{ d \notin \qv(C_0) }{
        C_0; \borr{d}{C_1} \leftrightarrow \borr{d}{C_0; C_1}
    }

    \inferrule*[right=\RuleSwap]{ \comm{C_1}{C_2} }{
        C_1; C_2 \leftrightarrow C_2; C_1
    }

    \inferrule*[right=\RuleBorrRight]{ d \notin \qv(C_1) }{
        \borr{d}{C_0}; C_1 \leftrightarrow \borr{d}{C_0; C_1}
    }

    \inferrule*[right=\RuleAssoc]{ }{
        (C_1; C_2); C_3 \leftrightarrow C_1; (C_2; C_3)
    }

    \inferrule*[right=\RuleBorrComm]{ }{
        \borr{d_1}{\borr{d_2}{C}}
        \leftrightarrow
        \borr{d_2}{\borr{d_1}{C}}
    }

\end{mathpar}
\endgroup
Here, $\comm{C_1}{C_2}$ means that $C_1$ and $C_2$ are \emph{disjoint}, that is, they share no quantum variables.
We write $\leftrightarrow$ for a pair of transitions allowed in both directions.
We write $C \topotrans C'$ for one topological transition and
$C \topotrans^{*} C'$ for its reflexive--transitive closure. Since all
topological rules are bidirectional, $\topotrans^{*}$ is symmetric and hence
an equivalence relation.

\paragraph{Borrowing transition}
The following rule models the actual borrowing step: it substitutes the required dirty ancilla $d$ with an idle qubit $q$ (i.e., $q \notin \qv(C) \cup \qa(C)$), called the \emph{borrowed qubit}, which may be either a working qubit or another ancilla:
\[
    \inferrule*[right=\RuleSubst]{
        q \notin \qv(C) \cup \qa(C)
    }{
        \borr{d}{C} \to C[q/d]
    }
\]

As an example, consider the following transition sequence, where the highlighted subcircuit is the result of the previous transition. Assume that $\qv(C_1) = \{q_2, d\}$, $\qv(C_2) = \{q_1, q_3\}$, and
$\qv(C_3) = \{q_2, d\}$.
\begin{align*}
\borr{d}{C_1; C_2; C_3} 
&\xrightarrow{\RuleSwap} \borr{d}{C_1; \blue{C_3; C_2}} 
& (\comm{C_2}{C_3}) \\
&\xrightarrow{\RuleBorrRight} \textbf{borrow}[d]\{C_1; \blue{C_3\}; C_2} 
& (d \notin \qv(C_2)) \\
&\xrightarrow{\RuleSubst}     \blue{C_1[q_3/d]; C_3[q_3/d]}; C_2
& (q_3 \notin \qv(C_1; C_3)).
\end{align*}

The circuit is progressively rearranged by topological transition rules until the borrowing step becomes applicable. This sequence is also illustrated in \cref{fig:trans-rule-case}.

\begin{figure}[htb]
    \centering
    \includegraphics[width=.8\linewidth]{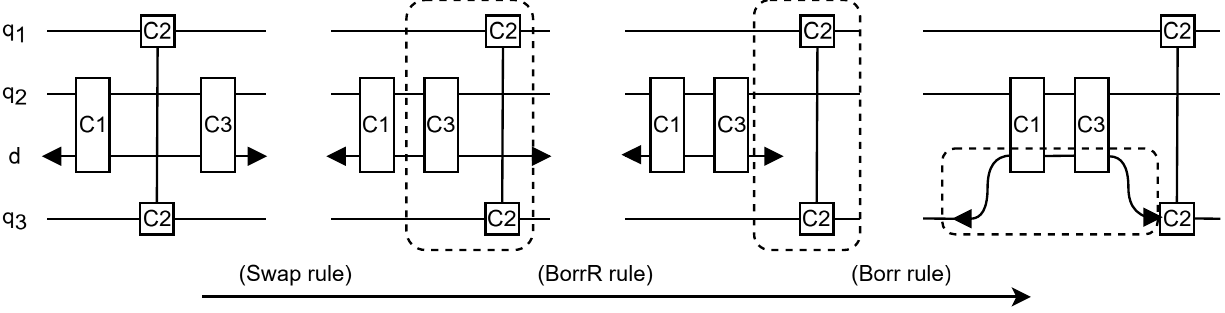}
    \caption{A circuit-level illustration of the transition sequence from left to right. Each labeled box represents a subcircuit. Boxes with the same label connected by a vertical line together denote a single subcircuit, indicating the quantum variables on which it acts. A dashed box denotes the subcircuit produced by the preceding transition.}
    \label{fig:trans-rule-case}
\end{figure}

\begin{definition}[Dirty-qubit borrowing problem]
\label{def:borrow-problem}
Given a safe circuit $C$, the dirty-qubit borrowing problem is to find a sequence of transitions from $C$ to $C'$ such that the width $|\argsop(C')|$ is minimal.
\end{definition}
The only way to reduce the width is through~\RuleSubst. We additionally assume that, in the following sections of this paper, the borrowed qubit $q$ is not a fresh name outside the initial circuit, because simply converting an ancilla into a new working qubit is trivial but does not help reduce the overall circuit width. Therefore, minimizing the width is equivalent to maximizing the number of applications of~\RuleSubst.

\begin{proposition}[Semantics preservation]\label{prop:semantics-preserving-tarnsitions}
For any safe circuit $C$ such that $C \to^{*} C'$, we have
\[
  \sem{C'} = \sem{C}.
\]
\end{proposition}

\paragraph{Why joint borrowing is non-trivial.}
Although each \RuleSubst~step is local, borrowing opportunities can interact,
requiring topological transitions between steps, as
\Cref{fig:frozen_wider} illustrates. In particular, exposing one borrowing
opportunity may block another, so the useful topological arrangement can
depend on earlier borrowing choices. We discuss this interaction in
\Cref{sec:tractable-case}.

\begin{figure}[htb]
  \centering
  \includegraphics[width=\linewidth]{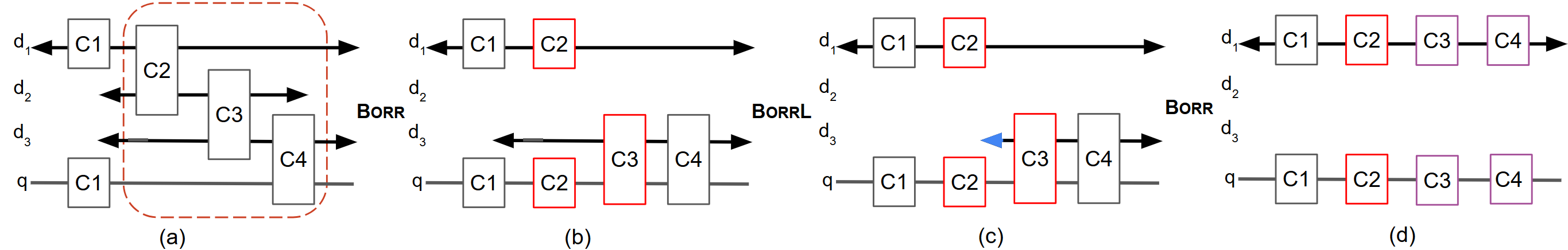}
  \caption{Interleaving topology and borrowing. Borrowing $q$ for $d_2$ first
  yields width 3 in (b); a topological transition then exposes $d_3$ in (c),
  enabling another borrowing to reach width 2 in (d).}
  \label{fig:frozen_wider}
\end{figure}

\section{Endpoint DAG Model}\label{sec:heuristic:graph}\label{sec:endpointdag}

By analyzing the structural and topological transition rules, we observe that \RuleSwap~ is the only one in structural and topological transitions that can modify gate order. Particularly, gates sharing a common qubit cannot swap, while gates on disjoint qubits can.
This observation motivates representing the constraints on gate order with a directed graph, which effectively quotients out the equivalence induced by structural and topological rules, allowing us to focus solely on the borrowing transitions.

\subsection{Definition}\label{sec:endpointdag:def}

To formalize this intuition, we associate with each circuit $C$ an \emph{endpoint DAG model} $\T(C) = (Q, Q_D, V, E)$, which encodes qubits, gates (i.e., operations), and their topological relationships in an easily visualized graph. The construction is as follows:
\begin{itemize}
  \item $Q$ is the set of all qubit identifiers, i.e., $\qv(C) \cup \qa(C)$.
  \item $Q_D \subseteq Q$ marks the dirty ancilla qubits, i.e., $\qa(C)$.
  \item $V$ is a finite set containing three types of nodes:
  \begin{itemize}
      \item $\In(q)$ and $\Out(q)$: one pair for each qubit $q \in Q$, representing the logical entry and exit points of $q$ in the circuit.
      \item $\Op(g)$: one for each quantum gate $g$ in $C$. For simplicity, we use ``node $g$'' to refer to ``node $\Op(g)$'' when there is no ambiguity, but $\In$ and $\Out$ are never omitted.
  \end{itemize}
  \item $E$ is a multi-set of directed edges over $V$: for each qubit $q \in Q$, let $g_1, g_2, \cdots , g_l$ be the gates acting on $q$ in program order. Then we add edges of the path
  \[(\In(q), g_1,  g_2, \cdots, g_l, \Out(q))\]
  to $E$, representing the temporal flow of gates along $q$.
\end{itemize}

This construction yields a directed acyclic graph (DAG) that preserves all gate-order constraints relevant to \RuleSwap, while abstracting away the nested-structure details and the gate semantics. For example,
the first three circuits in Fig.~\ref{fig:trans-rule-case} all correspond to the same
endpoint DAG model, illustrated in Fig.~\ref{fig:endpointmodel}.

\begin{figure}
    \centering
    \includegraphics[width=0.7\linewidth]{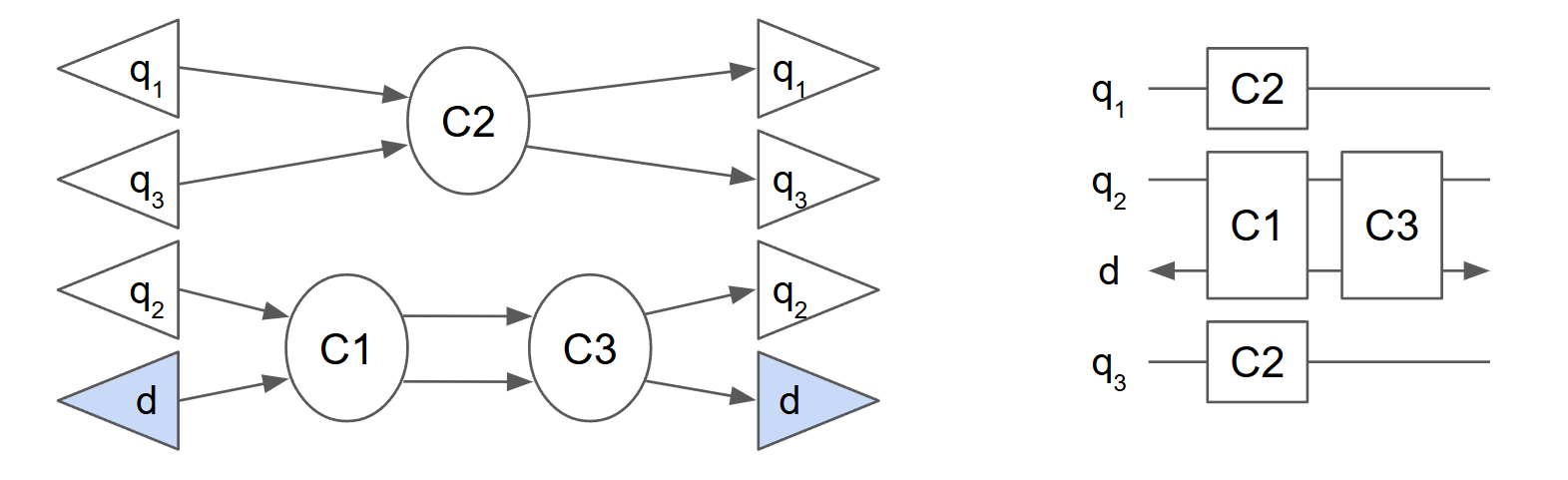}
    \caption{The Endpoint DAG Model as a canonical abstraction.
    The left panel shows the corresponding formal multi-digraph $(V, E)$,
    where each gate is a single node and edges track per-qubit temporal flow.
    The right panel shows a ``de-structured'' circuit view that reflects the
    model's indifference to syntactic nesting and borrowing scopes.
    Identical identifiers (e.g., $C_2$) represent a single multi-qubit gate
    acting on non-adjacent qubits ($q_1$ and $q_3$).
    }
    \label{fig:endpointmodel}
\end{figure}

The endpoint DAG model serves as a canonical representative for an entire
equivalence class of circuits under topological and structural transitions.
As illustrated in Figure~\ref{fig:endpointmodel} (right), our visualization of the model intentionally
omits syntactic details such as the boundaries of \texttt{borrow} statements or
the hierarchical nesting of subcircuits.
We draw a left-aligned circuit by default, i.e., each gate is placed as far to the left as possible. However, the reader should understand that the orders of $C_1$ and $C_2$,
and $C_2$ and $C_3$, are not fixed, since they do not have edges connecting them.
Furthermore, multi-qubit gates are treated as atomic units. In Figure~\ref{fig:endpointmodel}, the
two rectangles labelled $C_2$ represent a single gate acting on qubits $q_1$ and
$q_3$. Finally, we do not explicitly draw $\In$ and $\Out$ nodes,
since they are constructed only because we need the edges before the first operation
and after the last operation to represent the opportunity for borrowing.
We later denote the first and the last operation node of qubit $q$ by $\fst(q)$ and $\lst(q)$,
respectively.

Not every quadruple $(Q, Q_D, V, E)$ corresponds to a valid quantum circuit. To characterize exactly which endpoint DAGs can arise from a quantum circuit, we introduce a set of \emph{validity conditions}:
\begin{enumerate}
    \item \emph{Acyclic}: The multi-digraph $(V, E)$ contains no cycles, corresponding to the fact that quantum circuits do not allow loops.
    \item \emph{Balanced}: For all operation nodes $v = \Op(g)$, their in-degree $\deg^{-} (v)$ equals their out-degree $\deg^{+}(v)$, equal to the arity of $g$. This captures the fact that quantum gates neither create nor destroy
    qubits, but only transform them.
    \item \emph{Input/output consistency}: For each qubit $q \in Q$:
    \begin{itemize}
        \item $\In(q)$ has no incoming edges and $\Out(q)$ has no outgoing edges;
        \item there is a distinguished directed path $P_q$ from $\In(q)$ to $\Out(q)$.
    \end{itemize}
    Moreover, the edge multiset $E$ is the multiset-disjoint union of these paths:
    \[
        E = \biguplus_{q\in Q} E(P_q).
    \]
    This ensures that each qubit has well-defined starting and ending points, and
    that each edge occurrence belongs to exactly one qubit line.
\end{enumerate}

The power of this abstraction lies in its ability to bridge syntax and
structure. The following lemma confirms that the endpoint DAG model is not just
a visualization, but a rigorous foundation for circuit analysis. By quotienting
out the syntactic hierarchy, it establishes a bijective relationship between
valid DAGs and structural equivalence classes of circuits.
(The full proof is deferred to Appendix~\ref{app:correctDAG}.)

\begin{lemma}\label{lem:dag-properties}
The endpoint DAG model possesses two fundamental properties:
\begin{itemize}
    \item \textbf{Invariance:} If $C_1 \topotrans^{*} C_2$, then their DAG models are identical: $\T(C_1) = \T(C_2)$.
    \item \textbf{Validity:} An endpoint DAG $G$ corresponds to a quantum
circuit if and only if it satisfies the three validity conditions.
\end{itemize}
\end{lemma}

\begin{figure}[tb]
    \centering
    \includegraphics[width=1\linewidth]{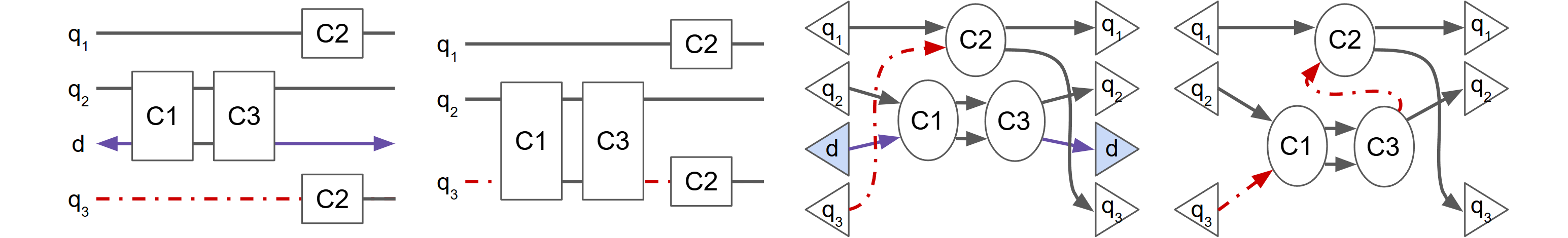}
    \caption{Borrowing as edge-splicing. In the DAG, substituting $d$ with
      $q_3$ cuts the idle edge $(\In(q_3), C_2)$ and splices $d$'s path into the
      gap (in red). The internal edge $(C_1, C_3)$ is preserved,
      while edges of endpoints of $d$ (in purple) are discarded.}
    \label{fig:dagborrow}
\end{figure}

\subsection{Borrowing as Edge-Splicing}
Having abstracted away topological reordering, we now examine how borrowing
transitions mathematically manifest in this graph model.

Consider the borrowing transition shown in Figure~\ref{fig:dagborrow}, where
the dirty ancilla $d$ is substituted by the borrowed qubit $q_3$ (same as in
Fig.~\ref{fig:trans-rule-case}).
When the substitution $C[q_3/d]$ occurs, the operations previously acting on
$d$ ($C_1$ and $C_3$) are forced to execute on $q_3$ before $C_2$.
From a graph-theoretic perspective, this transformation is
remarkably elegant: it simply cuts the idle edge $e = (\In(q_3), C_2)$ and
redirects the flow through $d$'s operations, creating new edges $(\In(q_3),
C_1)$ and $(C_3, C_2)$.
Since $d$ is no
longer a standalone qubit after the transition, its boundary nodes $\In(d),
\Out(d)$ and their terminal edges are naturally discarded.

Crucially, the internal structure of $d$'s path remains untouched as $(C_1, C_3)$, shown in Figure~\ref{fig:dagborrow}.
The causal execution order between $C_1$ and $C_3$ was already fixed by their
sequential action on $d$; the borrowing transition merely re-routes the
entry and exit points of this block to align with the borrowed qubit.
This geometric transformation allows us to model borrowing purely as a graph-level operation as follows:

\begin{definition}
An \emph{edge-splicing operation} on an endpoint DAG model $G = (Q, Q_D, V, E)$ is parameterized by an edge $e = (u,v) \in E$ and a dirty ancilla $d \in Q_D$,
producing a new quadruple
$G[e \triangleleft d] = (Q',Q'_D,V',E')$ defined as follows:
\[
\begin{cases}
Q' = Q \setminus \{d\}, \qquad Q_D' = Q_D \setminus \{d\},\\
V' = V \setminus \{\In(d),\Out(d)\},\\
E' = E \cup \{(u, \fst(d)), (\lst(d), v)\} \setminus \{(u,v), (\In(d),\fst(d)), (\lst(d), \Out(d)) \}
\end{cases}
\]

Intuitively, this operation removes the path corresponding to the
dirty qubit $d$ and splices it into the edge $e$.
\end{definition}

A natural question now arises: when is a borrowing step feasible?
A straightforward approach would be to check whether the resulting graph
$G[e \triangleleft d]$ still satisfies the validity conditions.

However, we show that feasibility can be characterized by a much simpler
unreachable condition, which also establishes the soundness of borrowing
in the endpoint DAG model.

\begin{definition}[Borrowing Step]\label{def:borrowing_dag}
Let $C$ be a circuit and $\T(C)=G=(Q,Q_D,V,E)$ its endpoint DAG model.  
For any edge $(u,v)\in E$ and dirty ancilla $d\in Q_D$,
the borrowing step is written in inference form:
\[
\frac{(u,v)\in E \qquad d\in Q_D \qquad \fst(d) \not\leadsto u \qquad v \not\leadsto \lst(d)}
     {G \rightarrow G[{(u,v) \triangleleft d}]}.
\]
Here, $G[(u,v) \triangleleft d]$ is the edge-splicing operation defined above.
\end{definition}

The necessity of the unreachable condition is straightforward: splicing
introduces two new directed edges, $(u, \fst(d))$ and $(\lst(d), v)$. If the
original graph contained a path $\fst(d) \leadsto u$ or $v \leadsto
\lst(d)$, adding the respective new edge would immediately create a cycle.
However, the sufficiency proof is more involved and is presented in Lemma~\ref{lem:splicing-validity} (Appendix~\ref{app:correctDAG}).

Crucially, this purely graphical borrowing step is perfectly isomorphic to
the borrowing transition ($\RuleSubst$) in the circuit model.
We formalize this as the following step
correspondence theorem (full proof provided in Appendix~\ref{app:correctDAG}):
      
\begin{theorem}[Step correspondence]
\label{thm:step-correspondence}
Let $C$ be a circuit and $\T(C)=G$ its endpoint DAG model.
\begin{itemize}
    \item (Soundness) If $G \to G'$ by an edge-splicing step,
    then there exist circuits $C_0$ and $C'$ s.t.:
    \[ C \topotrans^{*} C_0 \quad \text{and} \quad C_0 \xrightarrow{\RuleSubst} C' \quad \text{and} \quad \T(C') = G'. \]
    \item (Completeness) If $C \to C'$ by a borrowing transition,
    then $$\T(C) \to \T(C')$$ by an edge-splicing step.
\end{itemize}
\end{theorem}

\begin{proof}[Proof sketch]
\emph{(Soundness)}:
Suppose $G'=G[e\triangleleft d]$. We construct the target circuit $C'$
by specifying the parameters $q$ and $d$
of the borrowing transition. From the edge $e=(u,v)$, we
identify the borrowed qubit $q$ in $C$. The acyclicity of $G'$ mathematically
guarantees that, using \RuleBorrLeft, \RuleBorrRight, \RuleBorrComm,
\RuleAssoc, and \RuleSwap, the original circuit $C$ can be topologically
rearranged to $C_0$, i.e., $C\topotrans^*C_0$, such that all operations of
$d$ are sequentially between $u$ and $v$, and $q$ is unrelated to the code
block of $d$.
This allows the borrowing transition $C_0\xrightarrow{\RuleSubst}C'$, verifying that $\T(C') = G'$.

\emph{(Completeness)}:
Conversely, we can reveal $e$ and $d$ from the borrowing transition of the circuit.
Then it suffices to verify that $e$ and $d$ satisfy the unreachable condition,
and $\T(C)[e\triangleleft d] = G'$.
\end{proof}

\paragraph{Solving the Borrowing Problem.}
Recall from Definition~\ref{def:borrow-problem} that the ultimate goal of the dirty-qubit
borrowing problem is to find a sequence of transitions that minimizes the
resulting circuit width $|\argsop(C')|$. By
Theorem~\ref{thm:step-correspondence}, this problem is rigorously reduced to
finding a sequence of valid edge-splicing steps on the Endpoint DAG that
maximally shrinks the set $Q$. This structural correspondence forms the
unified operational foundation for our complexity analysis
(\Cref{sec:hardness}) and the \bona\ scheduling algorithm
(\Cref{sec:heuristic}).

\section{Complexity results}\label{sec:hardness}
In this section, we present complexity results of the dirty-qubit borrowing problem.
The hardness result motivates the heuristic-driven scheduler presented in \Cref{sec:heuristic}.

\subsection{An Exactly Solvable, Linear-Time Case: Frozen Topology}
\label{sec:tractable-case}

Recall that the goal of dirty-qubit borrowing is to
find a transition sequence that
minimizes the total width of the resulting circuit. Before addressing the general problem, we first isolate a restricted subclass by \emph{freezing the circuit topology}, i.e., disallowing topological transitions. In this setting, the circuit’s structural hierarchy and borrowing-scope boundaries are fixed, removing one degree of freedom from the search space. Consequently, any transition sequence can only apply \RuleSubst, so each step reduces to choosing which qubit to borrow. This restriction makes the subclass exactly solvable, as explained below.

At first glance, finding the optimal sequence of borrowing substitutions
still appears combinatorial. However, we show that under a frozen topology,
this sequence can be determined purely by computing a structural metric. Let
$W(C)$ be the minimal width of circuit $C$ achievable under a frozen
topology. $W(C)$ can be inductively defined by a bottom-up structural
recursion:

\begin{itemize}
    \item Basic gates:
    \[ W(U[\overline{q}]) = |\overline{q}| \]
    The width required by a single gate is at least the number of qubits it acts on.

    \item Sequential composition:
    \[ W(C_1; C_2) = \max\bigl(|\qv(C_1; C_2)|, W(C_1), W(C_2)\bigr) \]
    Working qubits traversing the composition require a baseline width of $|\qv(C_1;
C_2)|$. Beyond these working qubits, the local dirty ancillas in $C_1$
and $C_2$ are temporally disjoint. Thus, the qubit wire released by ancillas at
the end of $C_1$ can be immediately reallocated to ancillas in $C_2$.
The $\max$ operation
guarantees sufficient width to accommodate the maximum width
requirement of the internal blocks as well as the shared working
qubits.

    \item Borrowing declaration:
    \[ W(\borr{d}{C}) = W(C) \]
    Declaring a dirty ancilla merely establishes a scope boundary for a
borrowed qubit that is already accounted for within $W(C)$. This declaration
does not increase the overall width of the
circuit.
\end{itemize}

This structural simplicity places the dirty-qubit borrowing problem
with frozen topology
strictly in the complexity class P, providing a stark contrast to the
general case.

\begin{theorem}\label{thm:frozen-exact}
For any valid circuit $C$, the minimal achievable width under a frozen
topology, $W(C)$, is computable in linear time $O(|C|)$.
\end{theorem}

The theorem follows immediately by evaluating the above structural
recurrence in a single bottom-up traversal of $C$.

\paragraph{Transition to Intractability.}
The linear-time exactness of $W(C)$ reveals that dirty-qubit borrowing is
simple when scopes are static, but \Cref{fig:frozen_wider} demonstrates that
this rigidity leaves optimization potential untapped. In particular, a
two-stage heuristic that first applies all topological rules to shrink scopes
and then solves the frozen borrowing problem cannot expose all opportunities
needed by an optimal sequence.

This counterexample reveals a key insight: \emph{the optimal topological
structure is dynamically generated during the borrowing process.}
Borrowing a qubit changes the dependency graph, enabling new
topological moves, which in turn expose further borrowing opportunities.
Finding the optimal sequence of these deeply interleaved transitions
elevates the general problem to NP-hardness, as we establish next.

\subsection{Computational Complexity: The Hardness of Reuse} \label{sec:np-hardness}

In this section, we analyze the computational complexity of the dirty-qubit borrowing problem (Definition~\ref{def:borrow-problem}).  By guessing every borrowing step, it is easy to show that the decision version of the dirty-qubit borrowing problem is in NP. In the following,
we show that even a simplified version, the \fr\ problem (deciding if all dirty ancillas can be completely mapped), is NP-complete. This implies that finding an optimal borrowing strategy is likely intractable for large-scale circuits.

 \begin{definition}[\fr] The \fr\ problem is defined as follows:
    \begin{itemize}
        \item \textbf{Input:} An Endpoint DAG Model $G = (Q, Q_D, V, E)$.
        \item \textbf{Output:} Whether all dirty ancillas can be completely mapped, i.e., whether the optimal width can be $|Q-Q_D|$.
    \end{itemize}
\end{definition}

To bridge the gap between logic (3-SAT) and circuit topology, we introduce an intermediate graph problem. 

\begin{definition}[Uncertain Directed Graph, UDG]
    A UDG $\mathcal{G}=(V,E)$ is a graph where some edges are "uncertain", i.e. with uncertain tail or head. Specifically, an edge with  uncertain tail $Ts, s\in V, T\subseteq V$ has a fixed head $s$ but its tail must be chosen from a candidate set $T$. An edge with  uncertain head $tS$ is defined similarly. A \textit{realization} $G \in \mathcal{G}$ is obtained by picking exactly one vertex from the candidate set for each uncertain edge, resulting in a (certain) directed graph.
\end{definition}

\begin{lemma} \label{lem:NPdagudg}
    Deciding whether a UDG has a DAG realization (\daginudg) is NP-hard.
\end{lemma}

\begin{proof}[Proof Insight]
    We reduce from 3-SAT. For a formula $\psi$, we construct a UDG where:
    \begin{enumerate}
        \item \textbf{Variable Selection:} Each literal is a vertex and for each variable $x_i$, an uncertain edge from $\{x_i,\bar x_i\}$ to $x_{\text{false}}$ represents the choice between setting $x_i$ to True or False.
        \item \textbf{Clause Verification:} Each clause is a vertex as well and an uncertain edge starting from it to its literals represents the choice of the literals that satisfies the clause. 
        \item \textbf{Conflict as Cycles:} There is an edge from $x_{\text{false}}$ to $C$ and an edge from $C$ to $C_j$ for each clause $C_j$. The key insight is that if a clause is not satisfied by the chosen variable assignments, the construction forces a directed cycle in the graph. A DAG realization exists if and only if there is a satisfying assignment.
    \end{enumerate}
    See Figure \ref{fig:3SAT} for an example.
\end{proof}
\begin{figure*}
	\centering
    \subcaptionbox{\label{fig:UDG}}{ \includegraphics[width=0.3\textwidth]{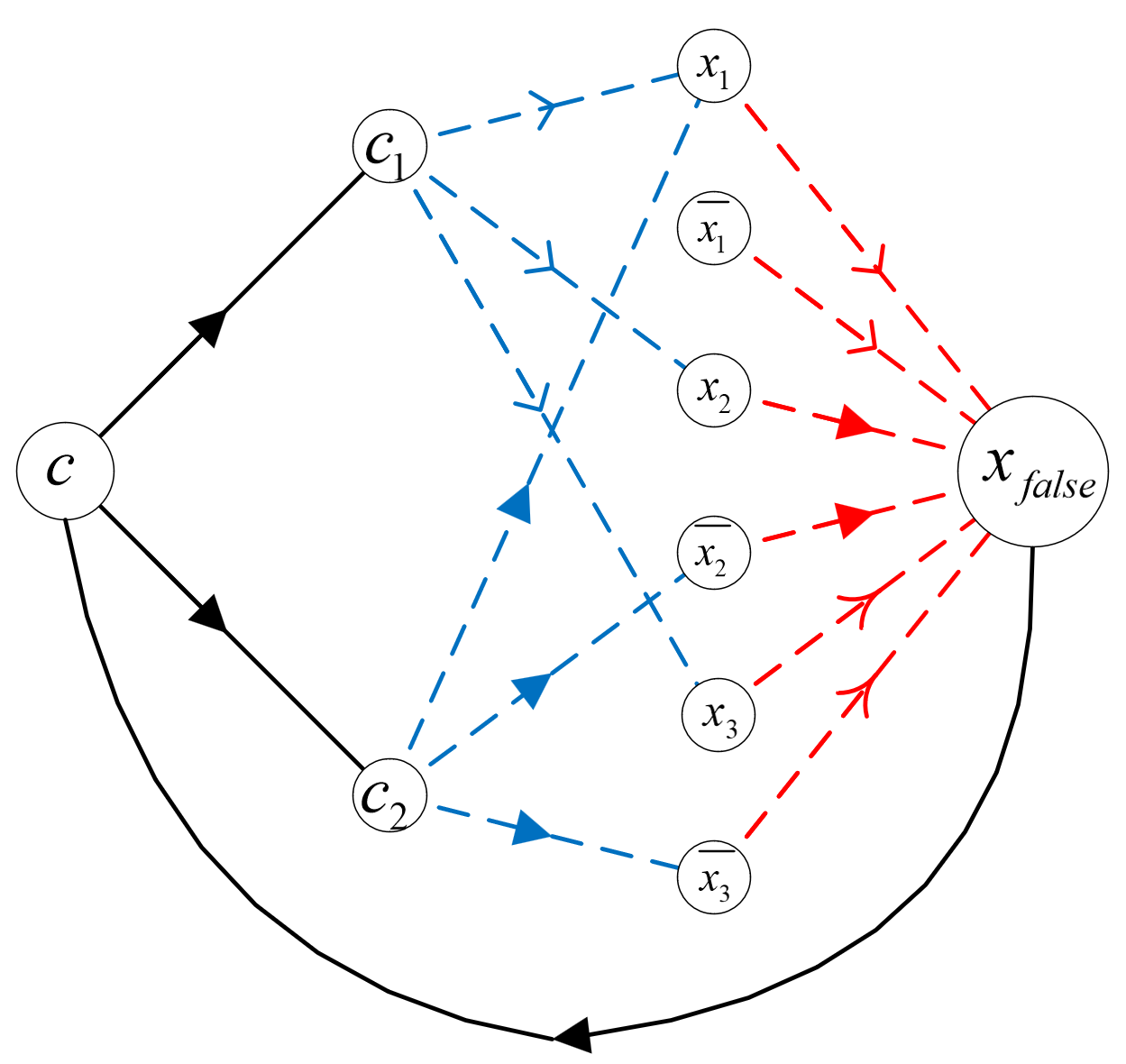}}
		\subcaptionbox{\label{fig:3SATT}}{\includegraphics[width=0.3\textwidth]{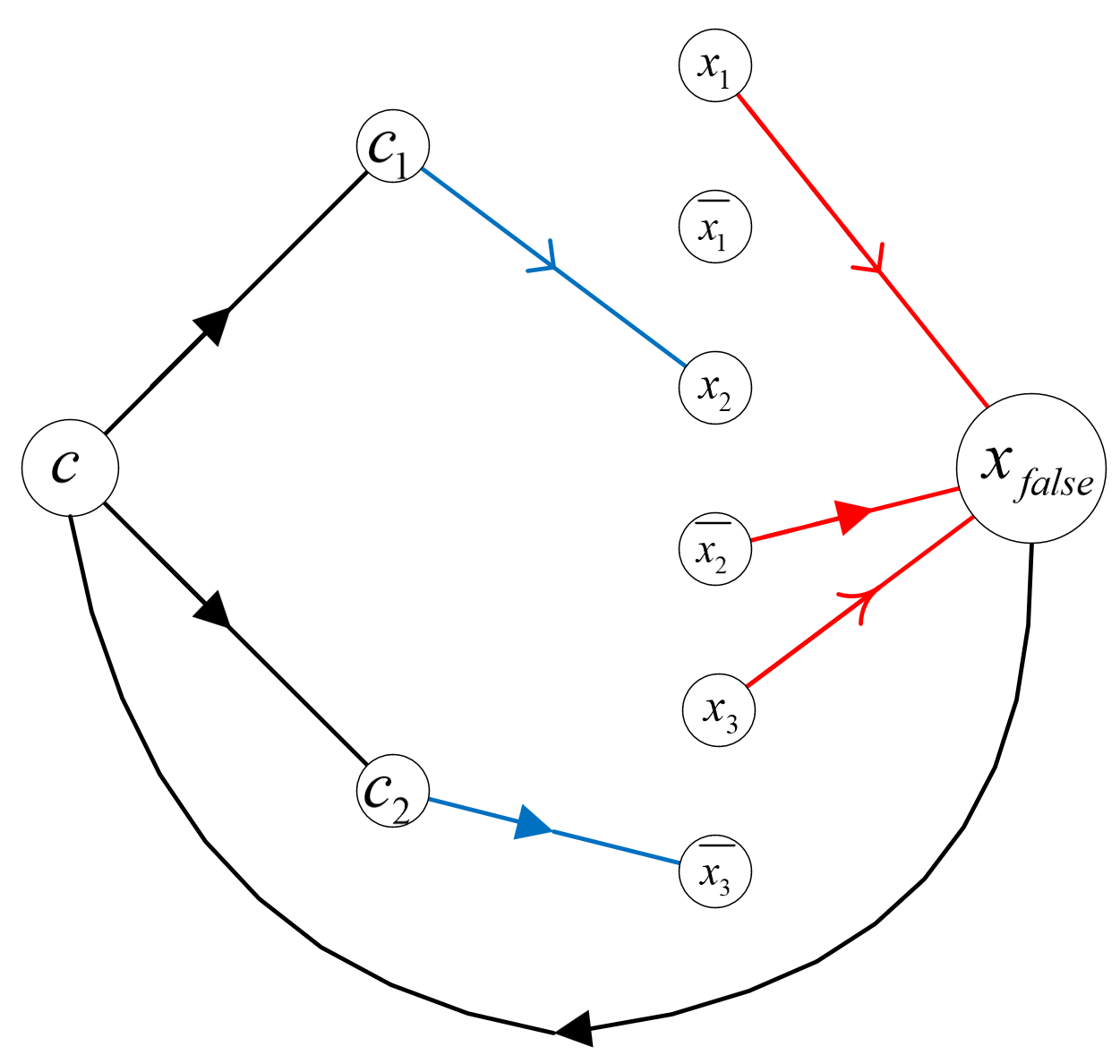}}
 \subcaptionbox{\label{fig:3SATF}}{\includegraphics[width=0.3\textwidth]{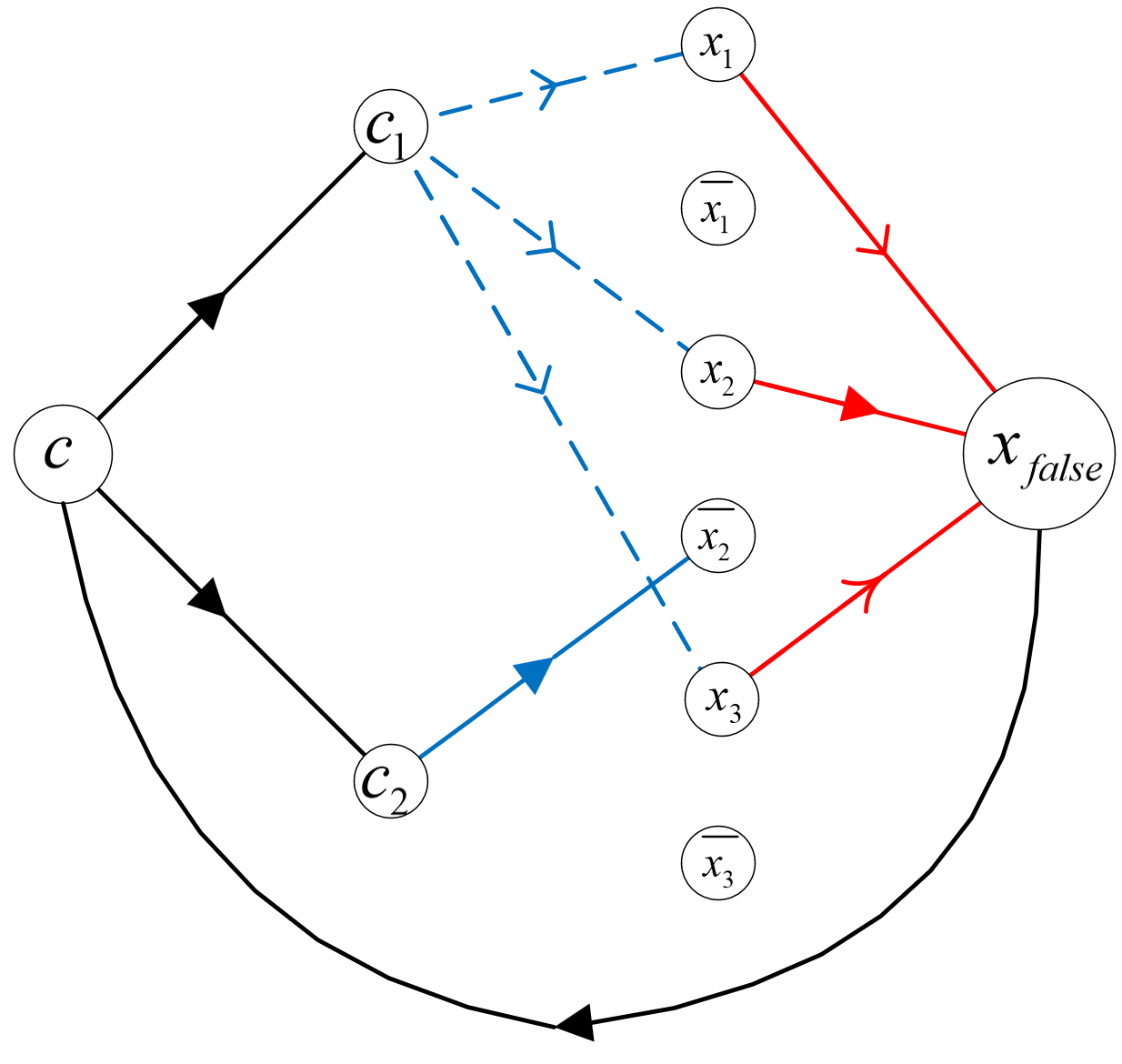}}
	\caption{Schematic diagram of the reduction from 3-SAT to \daginudg.  
    (a) The UDG constructed from a 3-SAT instance $(x_1\lor x_2\lor x_3)\land(x_1\lor \overline{x_2}\lor\overline{x_3})$. Each solid edge corresponds to a certain edge. Blue dashed edges are edges with an uncertain head, while red dashed edges are edges with uncertain tail. Different types of arrows are used to distinguish between distinct uncertain edges.
    (b) A DAG realization corresponding to the solution $(x_1,x_2,x_3)=(0,1,0)$.     
    (c) A partial realization corresponding to the non-solution $(x_1,x_2,x_3)=(0,0,0)$, which implies a cycle in the directed graph.} 
    % \vspace{-0.1cm}
	\label{fig:3SAT}
\end{figure*}

\begin{theorem} \label{thm:NPfr}
    \fr\ is NP-complete.
\end{theorem}
The full proof is deferred to Appendix~\ref{app:NP}, since even small UDG instances produce Endpoint DAGs too large to be readable in the main text.

\paragraph{The Complexity Landscape.}
The NP-completeness of \fr~immediately establishes
the hardness of our primary optimization objective.

\begin{corollary}\label{lem:maxreuse_nphard}
    The dirty-qubit borrowing problem (minimizing circuit width) is NP-hard.
\end{corollary}

To complete the theoretical picture, we also investigate the complexity of optimizing circuit depth,
or some lexicographical combination of both metrics.
Unfortunately, the hardness is pervasive across all meaningful formulations:
\begin{itemize}
    \item \textbf{Depth Minimization:} Minimizing depth alone is trivial---one simply performs
zero borrowing transitions. However, this is useless in practice.
    \item \textbf{Width-then-Depth Optimization:} Finding the optimal depth
\emph{subject to} an optimal width is strictly harder than finding the optimal width itself.
Thus, it is NP-hard.
    \item \textbf{Depth-then-Width Optimization:} Finding the optimal width
\emph{subject to} an optimal depth is also NP-hard via a reduction from \daginudg,
since one can trivially lock the optimal depth of the circuit by appending an ``extremely
long'' dummy working qubit that shares gates with every dirty ancilla.
\end{itemize}

Because every meaningful prioritization of width and depth encounters an
NP-hard combinatorial wall, searching for exact solutions in large-scale
circuits is computationally prohibitive. This intractable landscape
necessitates the design of a heuristic-driven scheduler, which we
introduce next.

\section{\bona: A Depth-Aware Borrowing Scheduler}\label{sec:heuristic}

We now present \bona, a depth-aware heuristic borrowing scheduling tool
built on the Endpoint DAG model. Previous sections demonstrated that
simultaneously optimizing width and depth is NP-hard, regardless of which
metric is chosen as the primary objective. \bona~is designed as a practical
scheduler that achieves a strong trade-off among the following objectives:
\begin{itemize}
  \item the resulting circuit width,
  \item the resulting circuit depth, and
  \item the runtime of the tool itself.
\end{itemize}

\subsection{Depth-Aware Borrowing Metric}

For a fixed dirty ancilla $d$, every successful borrowing step removes the
same dirty-ancilla wire and therefore yields the same immediate width
reduction, regardless of which feasible candidate edge is selected. When
multiple candidate edges satisfy the unreachable condition, the width
objective alone cannot distinguish among them. However, splicing $d$ into
each candidate edge introduces its own precedence constraints and may delay
different parts of the DAG. We quantify the resulting effect on circuit depth
by the actual depth change
\begin{equation}
    \depthdelta(e,d) := \operatorname{depth}(G[e\triangleleft d])-
    \operatorname{depth}(G).
\end{equation}
Figure~\ref{fig:alter} compares this quantity for two feasible choices.

\begin{figure}[h]
    \centering
    \includegraphics[width=1\linewidth]{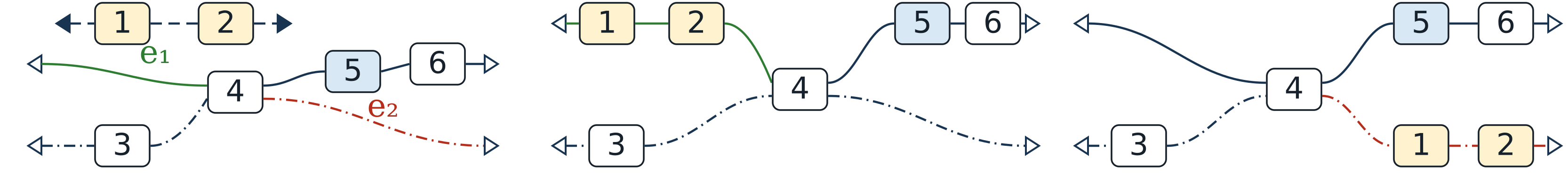}
        \caption{Different depth effects of two feasible borrowing choices.
(Left) The initial circuit contains a dirty ancilla $d$ with
$\fst(d)=1$ and $\lst(d)=2$, together with two
feasible candidate edges, $e_1$ and $e_2$, highlighted in green and red,
respectively.
(Middle) Splicing $d$ into $e_1$ delays operation $4$ and its successors,
producing a circuit of depth $5$, with $\depthdelta(e_1,d)=1$.
(Right) Splicing $d$ into $e_2$ instead extends the initially non-critical
bottom path and produces a circuit of depth $4$, with
$\depthdelta(e_2,d)=0$.}
    \label{fig:alter}
\end{figure}

Circuit depth is computed from the depths of its nodes. We write $\dep(v)$ for
the earliest possible layer of a node $v$, equivalently the length of a longest
path from any input node to $v$; the circuit depth is then
$\max_{v\in V}\dep(v)$. A standard topological traversal computes all node
depths~\cite{kahn1962topological}. Exactly evaluating $\depthdelta(e,d)$ would
require constructing the corresponding spliced graph and recomputing these
depths, which would be too costly for the
inner loop of a greedy scheduler. \bona\ instead estimates each candidate using a lightweight look-ahead cost derived from the current depth
assignment.

Suppose we attempt to borrow an idle period on some qubit, represented by an
edge $e = (u,v)$, for a dirty ancilla $d$. If we splice $d$ into $e$, the
operations of $d$ are forced to execute after $u$ and before $v$. This can
delay $\fst(d)$, $v$, or their successors; whether these local delays increase
the overall circuit depth depends on whether they affect a critical path.

We define a look-ahead cost function, $\estdepthdelta(e, d)$, to estimate
the local depth penalty incurred by this borrowing step:
\begin{equation}
    \estdepthdelta(e, d) = \max\bigl(\dep(u)+1-\dep(\fst(d)), 0\bigr) +
\max\bigl(\dep(\lst(d))+1-\dep(v), 0\bigr)
\end{equation}

\begin{figure}[h]
    \centering
    \includegraphics[width=1\linewidth]{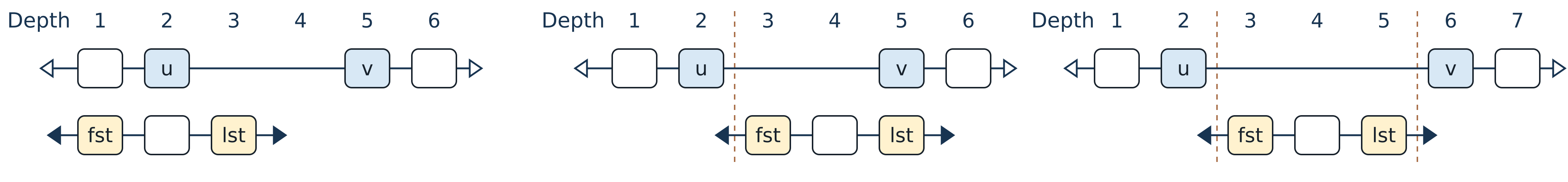}
        \caption{Visualizing the components of the depth-aware look-ahead cost
$\estdepthdelta(e, d)$.
    (Left) A dirty ancilla $d$ (spanning from $\fst(d)$ to $\lst(d)$) is
considered for borrowing onto idle edge $e = (u, v)$.
    (Middle) For $\fst(d)$ to follow $u$, its earliest depth becomes $\dep(u)
+ 1$, resulting in a local depth increase of $\max\bigl(\dep(u) + 1 -
\dep(\fst(d)), 0\bigr)$.
    (Right) For $v$ to follow $\lst(d)$, its earliest depth becomes
$\dep(\lst(d)) + 1$, resulting in an additional depth increase of
$\max\bigl(\dep(\lst(d)) + 1 - \dep(v), 0\bigr)$.
    The total heuristic cost is the sum of these two structural penalties.}
    \label{fig:depheu}
\end{figure}
This heuristic conservatively over-approximates the structural stretch required
to accommodate $d$ when the available edge $e$ is too short. It defines
\bona's depth-aware greedy objective: among feasible candidate edges, choose
one minimizing $\estdepthdelta$.

Specifically, in Figure~\ref{fig:alter},
$\estdepthdelta(e_1,d)=1 < \estdepthdelta(e_2,d)=2$, so \bona\ selects
$e_1$. However, the actual changes satisfy $\depthdelta(e_1,d)=1>
\depthdelta(e_2,d)=0$, so selecting $e_2$ yields the shallower circuit with depth
$4$ instead of $5$. Thus, $\estdepthdelta$ provides an inexpensive
depth-aware ranking but does not necessarily equal $\depthdelta$ or guarantee a
globally depth-optimal choice.

\subsection{The Borrowing Algorithm}

Because $\estdepthdelta$ is nonnegative, a zero-cost feasible edge attains the
minimum possible heuristic cost. Such a splice satisfies
$\dep(u)<\dep(\fst(d))$ and $\dep(\lst(d))<\dep(v)$ and preserves the current
depth assignment. Algorithm~\ref{alg:bona} treats this case as a fast path,
called \emph{soft borrowing}; all remaining cases are \emph{hard borrowing}.

\paragraph{Soft borrowing.}
The resource pool is designed to maintain active idle edges during the
depth-ordered sweep. When the sweep reaches $c=\dep(\fst(d))$, every pooled
edge $e=(u,v)$ satisfies $\dep(u)<c$. A covering query additionally requires
$\dep(\lst(d))<\dep(v)$ and, among matching edges, selects the minimum
$\dep(v)$ to preserve longer intervals for later ancillas. These inequalities
give zero cost and rule out both $\fst(d)\leadsto u$ and
$v\leadsto\lst(d)$, establishing the unreachable condition without a
reachability query. Implemented as a \texttt{SortedList} ordered by
right-endpoint depth~\cite{sortedcontainers_sortedlist}, the pool supports
covering queries and local updates in $O(\log |Q|)$ time.

\paragraph{Hard borrowing.}
If no covering edge exists, the hard branch invokes
$G.\text{QueryBestEdge}(d)$, which scans all edges, explicitly filters
$\fst(d)\not\leadsto u$ and $v\not\leadsto\lst(d)$, and minimizes
$\estdepthdelta(e,d)$. Precomputed descendant and ancestor
\texttt{set}s~\cite{python_set_type} make each filter check $O(1)$ and the
complete query $O(|E|)$. A successful hard splice invalidates $\dep$, so the
algorithm recomputes all depths in $O(|E|)$ time and restarts the sweep; if
the query fails, it adds $d$ to $\Skipped$ to avoid repeating the search.

\begin{algorithm}[H]
\small
\SetKwInOut{Input}{Input}
\SetKwInOut{Output}{Output}
\SetAlgorithmName{Algorithm}{alg:bona}{List of Algorithms}
\caption{Depth-Aware Heuristic Borrowing (\bona)}
\label{alg:bona}
\DontPrintSemicolon

\Input{Endpoint DAG $G=(Q, Q_D, V, E)$ with dirty ancillas $Q_D \subseteq
Q$.}
\Output{Updated $G$ with better width and depth}

$\Skipped \gets \emptyset$ \tcp*{Ancillas skipped after a failed global search}

\While{$Q_D \setminus \Skipped \neq \emptyset$}{
    $\dep \gets \text{ComputeDepth}(G)$ \tcp*{Via a topological sort}
    $\Pool \gets \emptyset$ \tcp*{Maintains valid idle edges crossing current
depth}

    \ForEach{depth $c \in [0, \text{MaxDepth}(G)]$}{
        \ForEach{dirty ancilla $d \in Q_D \setminus \Skipped$ with
        $\dep(\fst(d)) = c$}{
            $e_{soft} \gets \Pool.\text{QueryCoveringEdge}(d)$\;

            \eIf{$e_{soft} \neq \text{NULL}$}{
                $G \gets G[e_{soft} \triangleleft d]$ \tcp*{Soft Borrowing:
zero depth penalty}
                $\Pool.\text{Remove}(e_{\mathrm{soft}})$\;
            }{
                $e_{hard} \gets G.\text{QueryBestEdge}(d)$\;
                \eIf{$e_{hard} \neq \text{NULL}$}{
                    $G \gets G[e_{hard} \triangleleft d]$ \tcp*{Hard
Borrowing: invalidates $\dep$}
                    \textbf{goto} Line 2 \tcp*{Topology changed, recompute
$\dep$}
                }
                {
                    \tcp{If $e_{hard}$ is NULL, $d$ remains unmapped and is
skipped}
                    $\Skipped \gets \Skipped \cup \{d\}$
                }
            }
        }

        \ForEach{operation node $v \in V$ with $\dep(v) = c$}{
            \ForEach{incoming edge $e' = (w, v) \in E$ to node $v$}{
$Pool.\text{Remove}(e')$ }
            \ForEach{outgoing edge $e = (v, u) \in E$ from node $v$}{
$Pool.\text{Add}(e)$ }
        }
    }

    \tcp{Terminate if a full round yields no successful mapping}
    \If{no $G[e \triangleleft d]$ was applied in this round}{ \textbf{break}
}
}
\Return $G$\;
\end{algorithm}

\begin{figure}[H]
    \centering
    \includegraphics[width=1\linewidth]{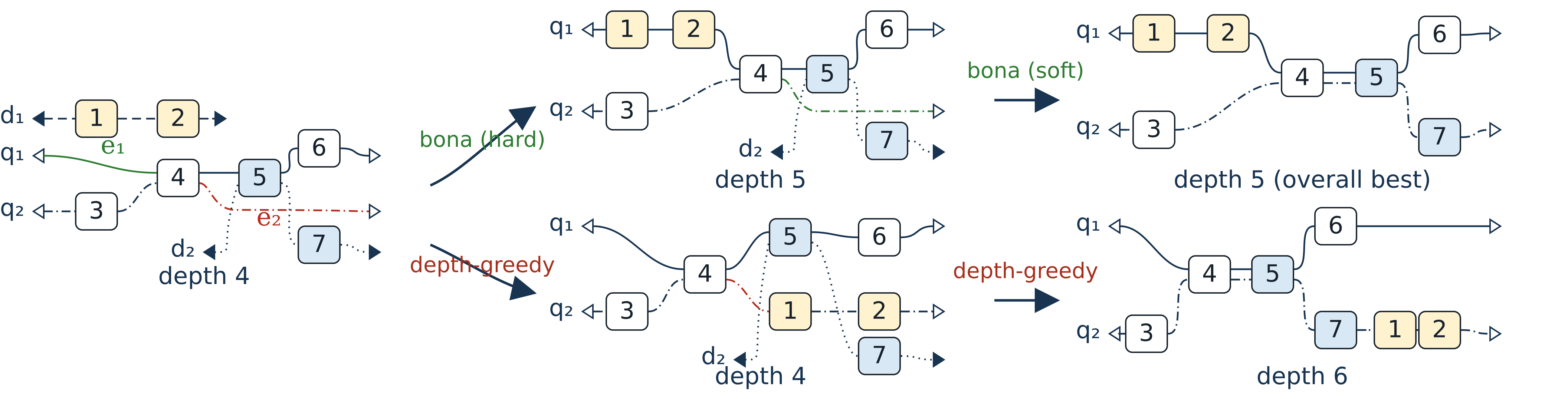}
    \caption{Borrowing schedules for the circuit in
Figure~\ref{fig:alter}, augmented with a second dirty ancilla $d_2$:
\bona\ reaches the overall-best depth $5$, whereas greedily minimizing
$\depthdelta$ at each step reaches depth $6$.}
    \label{fig:heuristic}
\end{figure}

In the upper branch of Figure~\ref{fig:heuristic}, Algorithm~\ref{alg:bona}
first processes $d_1$. At this point, $e_1$ is already in the resource pool,
but its right endpoint has the same depth as $\lst(d_1)$. It therefore fails
the strict covering condition $\dep(\lst(d_1))<\dep(v)$, and
$\text{QueryCoveringEdge}$ returns \text{NULL}. The algorithm enters the hard
branch, where $\text{QueryBestEdge}$ selects $e_1$ because
$\estdepthdelta(e_1,d_1)=1$ is minimal among the feasible candidate edges.
This splice changes the circuit depth from $4$ to $5$ and terminates the
current traversal. When the traversal restarts and reaches $d_2$, $e_2$ is in
the pool and is returned by $\text{QueryCoveringEdge}$. The resulting soft borrowing has
$\estdepthdelta(e_2,d_2)=0$, requires no depth recomputation, and leaves the
final depth at $5$.

If each dirty ancilla were instead assigned greedily to a feasible edge that
minimizes the actual depth increase $\depthdelta$ in the current DAG, the
first choice would be $e_2$, since $\depthdelta(e_2,d_1)=0<
\depthdelta(e_1,d_1)=1$. The circuit would initially remain at depth $4$.
However, after this splice changes the available idle intervals, the minimum
feasible choice that eliminates $d_2$ increases the final depth to $6$.

This outcome reflects the dynamic-topology difficulty identified in
Section~\ref{sec:hardness}: each borrowing step changes the topology, and
hence the candidate edges, seen by later decisions. Rather than materializing
every candidate splice to compute $\depthdelta$, \bona\ uses a cheap,
structure-sensitive estimate that exposes a zero-cost fast path. Soft
candidates are found through the resource pool without recomputing depths;
graph-wide search is reserved for hard borrowing, with depth recomputation
only after a successful hard splice. This design makes \bona\ efficient
without substantially compromising the resulting width--depth trade-off, as
evaluated in Section~\ref{sec:evaluation}.

Although \bona\ does not guarantee a globally optimal schedule, the following
theorem establishes that its output corresponds to a valid circuit
semantically equivalent to the input.

\begin{theorem}[Soundness of \bona]\label{thm:bona-soundness}
Let $C$ be a safe circuit, and let $G=\T(C)$ be given as input to
Algorithm~\ref{alg:bona}. If Algorithm~\ref{alg:bona} outputs $G'$, then
$G'$ is a valid endpoint DAG. Moreover, there exists a circuit $C'$ such
that $G'=\T(C')$ and $C \to^* C'$, and $\sem{C}=\sem{C'}$.
\end{theorem}

\begin{proof}
Algorithm~\ref{alg:bona} changes the graph only through successful
assignments $G\gets G[e\triangleleft d]$; skipped ancillas leave it
unchanged. Consider one such assignment with $e=(u,v)$. In the soft branch,
the pool invariant and covering query give
$\dep(u)<\dep(\fst(d))$ and $\dep(\lst(d))<\dep(v)$, respectively. Since
reachability in a DAG strictly increases depth, these inequalities rule out
both $\fst(d)\leadsto u$ and $v\leadsto\lst(d)$. In the hard branch,
$\text{QueryBestEdge}$ explicitly retains only edges satisfying the same two
unreachable relations. Thus every applied splice satisfies the unreachable
condition of Definition~\ref{def:borrowing_dag}.

By Theorem~\ref{thm:step-correspondence}, each update
corresponds to a finite sequence of circuit transitions ending in a circuit
whose endpoint DAG is the updated graph. Induction over the successful
updates, including the empty sequence, yields a circuit $C'$ such that
$G'=\T(C')$ and $C\to^*C'$. By Proposition~\ref{prop:semantics-preserving-tarnsitions}, $\sem{C}=\sem{C'}$.
\end{proof}

\subsection{Integration with Clean-Ancilla Optimizers}

While \bona\ specifically targets the dirty-qubit borrowing problem,
practical quantum algorithms frequently employ a mix of both clean and dirty
ancillas. To build an end-to-end compiler pipeline, \bona\ is designed to
seamlessly couple with any state-of-the-art clean-ancilla optimizer (e.g.,
\recycle~\cite{jiang2024recycling} or \quantinuum~\cite{quantinuum}).

To orchestrate hybrid resources effectively, we establish a strict
\emph{Clean-then-Dirty} execution order, driven by the asymmetry in resource
flexibility. Clean ancillas carry rigid constraints: they can only recycle
idle periods through serial concatenation (strictly before or after another
clean ancilla) and cannot borrow temporary segments from existing working
qubits. Dirty ancillas, conversely, are universally flexible and can be
spliced into any edge---regardless of whether the borrowed qubit is a working one or
another ancilla, and whether the insertion point is internal or at the
endpoints. In resource allocation, satisfying the most constrained requests
first naturally yields better overall utilization.

Following this principle, the pipeline's first pass invokes a clean-ancilla
optimizer, treating dirty-borrowing scopes as fixed operations. This pass
minimizes the number of dedicated clean wires. Once allocated, \bona\ simply
treats these resulting clean wires as standard working qubits, exposing their
idle periods as available edges for dirty-qubit borrowing.

\paragraph{Implicit Clean-Dirty Interleaving.}
Because dirty ancillas are scheduled second, they can freely embed themselves
into idle gaps between clean ancillas, including the endpoints.
Consequently, complex interleaved
Clean-Dirty chains are implicitly constructed by Algorithm~\ref{alg:bona}
without requiring specialized joint-scheduling logic. Ultimately, this
modular design allows \bona\ to focus purely on its graph-theoretic core
while leveraging upstream tools to maximize overall resource utilization.

\section{Case Studies}\label{sec:evaluation}

We implement \bona~ in Python and evaluate it to answer the following research questions (RQs):
\begin{itemize}
    \item \textbf{RQ1 (Effectiveness):} How effective is \bona~ in reducing circuit width while limiting the overhead in circuit depth?
    \item \textbf{RQ2 (Advantage):} What scheduling advantages do (safely used) dirty ancillas offer over clean ancillas in terms of circuit width and depth?
    
    \item \textbf{RQ3 (Efficiency):} Does \bona~ run efficiently on large benchmark instances?
\end{itemize}

\subsection{Setup}

\paragraph{Benchmarks}
We evaluate \bona~ across three categories of benchmarks:
\begin{enumerate}
    \item \emph{Quantum algorithms}: 
    a parallel algorithm \emph{Parallel Quantum Walk (PQW)}~\cite{Zhang2024parallelquantum} and 
    \emph{Shor’s Algorithm}~\cite{Shor_1997} using the implementation of~\cite{factoring2np2_2017}.

    \item \emph{Component-level circuits}: \texttt{MCX} (Multi-controlled-$X$~\cite{gidney2015}), 
    \texttt{Incrementer}~\cite{gidney2015}, 
    \texttt{Grover}~\cite{grover}, 
    and \texttt{StatePreparation}~\cite{Low2024tradingtgatesdirty}.
    
    \item \emph{Parallel-composition scenarios}: 
    randomly selected RevLib subcircuits~\cite{wille2008revlib} stitched together under varying degrees of parallelism.
\end{enumerate}

\paragraph{Sources of ancillas}
In all benchmarks, ancillas come from either explicit allocations in the original design or from decomposing multi-controlled-$X$ gates ($\tMCX{n}$) into V-chain (clean) or M-chain (dirty) implementations before optimization. 

\paragraph{Allocation of ancillas}
Manually designed circuits are often already optimized for ancilla reuse, leaving few or no opportunities for further automatic optimization. Large quantum algorithms, however, are typically synthesized component by component. Although ancillas may be reused within each component by design, reuse opportunities across independently synthesized components are usually left unresolved, making cross-component reuse a natural target for automatic optimization.
To expose these opportunities without prematurely constraining the optimizer, the components are initially assigned disjoint ancilla sets. In our case studies, we evaluate \bona\ on circuits constructed in this manner and compare the resulting circuits with manually optimized implementations. Accordingly, ancillas are allocated in two ways.
\begin{itemize}
\item \emph{Max-ancilla version:} A new ancilla is allocated whenever needed, either by the original design or by the decomposition of $\tMCX{n}$ gates.

\item \emph{Manually optimized version:} ancilla reuse from the original design is preserved when available; otherwise, reuse is introduced by hand only for relatively direct opportunities (e.g., when one ancilla can be reused immediately after another use ends).
\end{itemize}

\paragraph{Baselines}
In our evaluation, we use two clean-ancilla reuse optimizers: \recycle~\cite{jiang2024recycling}, with its best-performing ``Greedy+Max0s'' heuristic, and \quantinuum~\cite{quantinuum}, using the implementation provided by~\cite{jiang2024recycling}.
When dirty-ancilla circuits also contain clean ancillas, we first apply \recycle\ to reuse the clean ancillas and then apply either \bona\ or \trivD\ to reuse the dirty ancillas. \trivD\ is a depth-aware greedy serial-reuse baseline that reuses a qubit only after the previous dirty ancilla has completely finished. For \textbf{RQ2}, we compare the \recycle+\bona\ pipeline on dirty-ancilla circuits with the two optimizers applied to the corresponding clean-ancilla circuits.

\paragraph{Metrics}
We evaluate three key circuit metrics: 
Width, 
Depth, and 
Clean/Dirty-ancilla count.
Circuit gates remain at the logical level; for example, Toffoli gates are not further decomposed into hardware-level gates.

\paragraph{Input-safety and output-soundness testing}
We evaluate input safety and output soundness through sampling-based testing, as exhaustive verification of both properties is computationally expensive.

\paragraph{Execution Environment}
All experiments were conducted in a Linux environment under WSL2 on a machine with an Intel Core i9-14900HX CPU and 64\,GB of RAM.

\subsection{Capturing Ancilla-usage Structure in Quantum Algorithms}

Many quantum algorithms make extensive use of ancilla qubits, making ancilla reuse important for controlling circuit width. However, their complex algorithmic structures often make such reuse difficult to perform manually. In this subsection, we evaluate \bona\ on two representative algorithms, parallel quantum walk and Shor's algorithm, to assess whether it can capture realistic ancilla-usage structures and exploit them for effective reuse.

\subsubsection{Parallel Quantum Walk}

Dirty ancillas may offer scheduling advantages in highly parallel circuits. 
Therefore, we evaluate \bona\ on the parallel quantum walk structure used in a parallel Hamiltonian simulation algorithm~\cite{Zhang2024parallelquantum}. 

\paragraph{Implementation details}
Specifically, we consider its core block-encoding construction for powers of a Hamiltonian based on parallel quantum walks (PQW). This construction arranges $r$ copies of a data-lookup oracle into $r$ parallel layers, yielding a highly parallel circuit with substantial ancilla requirements.
We instantiate the data-lookup oracle $O_H$ as in~\cite{Low2024tradingtgatesdirty}, where dirty ancillas are used as temporary workspace.
Since \cite{Zhang2024parallelquantum} provides only an algorithmic description, we implement the circuit following its algorithmic structure. The manually optimized version is constructed by hand and incorporates only relatively direct ancilla-reuse opportunities.

\paragraph{Experimental results}

As shown in~\cref{tab:pqw}, for the dirty-ancilla implementations, we first apply \recycle\ to reuse clean ancillas and then apply either \trivD\ or \bona\ to reuse dirty ancillas. \recycle+\trivD\ reduces width by 88\%--97\% and dirty-ancilla usage by 93\%--98\%, while incurring a depth overhead of 24\%--99\%. In comparison, \recycle+\bona\ reduces width further, by 92\%--99\%, and eliminates 99\%--100\% of dirty ancillas, leaving at most five, at the cost of a larger depth overhead of 101\%--330\% (\textbf{RQ1}).
With a preprocessed manual optimization, \recycle+\bona\ eliminates all dirty ancillas, reduces width by 89\%--97\%, and increases depth by only 1\%--18\%, while consistently achieving smaller width and depth than \recycle+\trivD\ (\textbf{RQ1}).

For the clean-ancilla implementations, \recycle\ and \quantinuum\ directly reuse clean ancillas. Although they achieve comparable width reductions of 91\%--98\%, they incur depth overheads of 119\%--528\%, and \recycle\ times out on \texttt{PQW70\_6}.
Thus, despite starting from deeper dirty-ancilla circuits, the dirty-ancilla pipelines achieve comparable width with substantially smaller depth after optimization at most scales (\textbf{RQ2}).

At the largest completed scales, excluding the \recycle\ preprocessing time, the running times for the four dirty-ancilla columns from left to right are 5327\,s for \trivD, 8301\,s for \bona, 10444\,s for \trivD, and 260\,s for \bona, respectively (\textbf{RQ3}). For the clean-ancilla versions, the running times at the largest completed scales are 8288\,s for \recycle and 19855\,s for \quantinuum.

\definecolor{BrickRed}{rgb}{0.8,0.1,0.1}
\definecolor{RoyalBlue}{rgb}{0.1,0.2,0.7}
\definecolor{DarkGray}{gray}{0.4}

\newcommand{\poschg}[1]{\textcolor{BrickRed}{#1}}
\newcommand{\negchg}[1]{\textcolor{RoyalBlue}{#1}}
\newcommand{\neuchg}[1]{\textcolor{DarkGray}{#1}}

\begin{table*}[!t]
  \centering
  \small
  \setlength{\abovecaptionskip}{2pt}
  \setlength{\belowcaptionskip}{0pt}

  % ==========================================================
  % Part 1: Width / Depth
  % ==========================================================
  \begin{minipage}{\linewidth}
    \centering
    \vspace{2pt}

    \resizebox{\linewidth}{!}{
      \begin{tabular}{lcccccccccc c lcccccc}
        \toprule
        & \multicolumn{10}{c}{\textbf{Dirty-Ancilla Implementation}} 
        & 
        & \multicolumn{6}{c}{\textbf{Clean-Ancilla Implementation}} \\
        \cmidrule(lr){2-11}\cmidrule(lr){13-18}
        \multirow{2}{*}{Circuit}
          & \multicolumn{2}{c}{\textbf{Original}}
          & \multicolumn{2}{c}{\textbf{\trivD}}
          & \multicolumn{2}{c}{\textbf{\bona}}
          & \multicolumn{2}{c}{\textbf{Manual+\trivD}}
          & \multicolumn{2}{c}{\textbf{Manual+\bona}}
          &
          & \multicolumn{2}{c}{\textbf{Original}}
          & \multicolumn{2}{c}{\textbf{\recycle}}
          & \multicolumn{2}{c}{\textbf{\quantinuum}} \\
        \cmidrule(lr){2-3}\cmidrule(lr){4-5}\cmidrule(lr){6-7}
        \cmidrule(lr){8-9}\cmidrule(lr){10-11}
        \cmidrule(lr){13-14}\cmidrule(lr){15-16}\cmidrule(lr){17-18}
          & Width& Depth
          & Width& Depth
          & Width& Depth
          & Width& Depth
          & Width& Depth
          &
          & Width& Depth
          & Width& Depth
          & Width& Depth \\
        \midrule

        PQW18\_2
          & 418 & 697
          & \makecell{51 \\ \negchg{-88\%}} & \makecell{891 \\ \poschg{+28\%}}
          & \makecell{\textbf{32} \\ \negchg{-92\%}} & \makecell{1402 \\ \poschg{+101\%}}
          & \makecell{48 \\ \negchg{-89\%}} & \makecell{811 \\ \poschg{+16\%}}
          & \makecell{46 \\ \negchg{-89\%}} & \makecell{\textbf{704} \\ \poschg{+1\%}}
          &
          & 418 & 651
          & \makecell{\textbf{34} \\ \negchg{-92\%}} & \makecell{1643 \\ \poschg{+152\%}}
          & \makecell{39 \\ \negchg{-91\%}} & \makecell{\textbf{1428} \\ \poschg{+119\%}} \\
        PQW42\_6
          & 1338 & 725
          & \makecell{152 \\ \negchg{-89\%}} & \makecell{1442 \\ \poschg{+99\%}}
          & \makecell{\textbf{93} \\ \negchg{-93\%}} & \makecell{2014 \\ \poschg{+178\%}}
          & \makecell{156 \\ \negchg{-88\%}} & \makecell{1025 \\ \poschg{+41\%}}
          & \makecell{150 \\ \negchg{-89\%}} & \makecell{\textbf{856} \\ \poschg{+18\%}}
          &
          & 1338 & 679
          & \makecell{\textbf{93} \\ \negchg{-93\%}} & \makecell{4167 \\ \poschg{+514\%}}
          & \makecell{104 \\ \negchg{-92\%}} & \makecell{\textbf{3030} \\ \poschg{+346\%}} \\
        PQW56\_6
          & 3860 & 2527
          & \makecell{212 \\ \negchg{-95\%}} & \makecell{3270 \\ \poschg{+29\%}}
          & \makecell{\textbf{111} \\ \negchg{-97\%}} & \makecell{8790 \\ \poschg{+248\%}}
          & \makecell{212 \\ \negchg{-95\%}} & \makecell{3839 \\ \poschg{+52\%}}
          & \makecell{202 \\ \negchg{-95\%}} & \makecell{\textbf{2751} \\ \poschg{+9\%}}
          &
          & 3860 & 2157
          & \makecell{\textbf{107} \\ \negchg{-97\%}} & \makecell{13542 \\ \poschg{+528\%}}
          & \makecell{162 \\ \negchg{-96\%}} & \makecell{\textbf{10492} \\ \poschg{+386\%}} \\
        PQW70\_6
          & 10294 & 7025
          & \makecell{299 \\ \negchg{-97\%}} & \makecell{8685 \\ \poschg{+24\%}}
          & \makecell{\textbf{126} \\ \negchg{-99\%}} & \makecell{30175 \\ \poschg{+330\%}}
          & \makecell{293 \\ \negchg{-97\%}} & \makecell{10731 \\ \poschg{+53\%}}
          & \makecell{280 \\ \negchg{-97\%}} & \makecell{\textbf{7358} \\ \poschg{+5\%}}
          &
          & 10294 & 5419
          & \makecell{- \\ \negchg{-}} & \makecell{- \\ \poschg{-}}
          & \makecell{\textbf{226} \\ \negchg{-98\%}} & \makecell{\textbf{29184} \\ \poschg{+439\%}} \\
        \bottomrule
      \end{tabular}
    }
  \end{minipage}

  \vspace{8pt}

  % ==========================================================
  % Part 2: Clean/Dirty Ancilla Usage
  % ==========================================================
  \begin{minipage}{\linewidth}
    \centering
    \vspace{2pt}

    \resizebox{\linewidth}{!}{
      \begin{tabular}{lccccc c lccc}
        \cmidrule(lr){2-6}\cmidrule(lr){8-10}
        
          & \textbf{Original}
          & \textbf{\trivD}
          & \textbf{\bona}
          & \textbf{Manual+\trivD}
          & \textbf{Manual+\bona}
          &
          & \textbf{Original}
          & \textbf{\recycle}
          & \textbf{\quantinuum} \\
        \cmidrule(lr){2-6}\cmidrule(lr){8-10}
          & C/D & C/D & C/D & C/D & C/D
          &
          & C/D & C/D & C/D \\
        \midrule

        PQW18\_2
          & 76/324
          & \makecell{\textbf{11}/22 \\ \negchg{-86\%} / \negchg{-93\%}}
          & \makecell{\textbf{11}/3 \\ \negchg{-86\%} / \negchg{-99\%}}
          & \makecell{28/2 \\ \negchg{-63\%} / \negchg{-99\%}}
          & \makecell{28/\textbf{0} \\ \negchg{-63\%} / \negchg{-100\%}}
          &
          & 400/0
          & \makecell{\textbf{16}/0 \\ \negchg{-96\%} / \neuchg{0\%}}
          & \makecell{21/0 \\ \negchg{-95\%} / \neuchg{0\%}} \\

        PQW42\_6
          & 324/972
          & \makecell{\textbf{51}/59 \\ \negchg{-84\%} / \negchg{-94\%}}
          & \makecell{\textbf{51}/\textbf{0} \\ \negchg{-84\%} / \negchg{-100\%}}
          & \makecell{108/6 \\ \negchg{-67\%} / \negchg{-99\%}}
          & \makecell{108/\textbf{0} \\ \negchg{-67\%} / \negchg{-100\%}}
          &
          & 1296/0
          & \makecell{\textbf{51}/0 \\ \negchg{-96\%} / \neuchg{0\%}}
          & \makecell{62/0 \\ \negchg{-95\%} / \neuchg{0\%}} \\

        PQW56\_6
          & 456/3348
          & \makecell{\textbf{51}/105 \\ \negchg{-89\%} / \negchg{-97\%}}
          & \makecell{\textbf{51}/4 \\ \negchg{-89\%} / \negchg{-100\%}}
          & \makecell{146/10 \\ \negchg{-68\%} / \negchg{-100\%}}
          & \makecell{146/\textbf{0} \\ \negchg{-68\%} / \negchg{-100\%}}
          &
          & 3804/0
          & \makecell{\textbf{51}/0 \\ \negchg{-99\%} / \neuchg{0\%}}
          & \makecell{106/0 \\ \negchg{-97\%} / \neuchg{0\%}} \\

        PQW70\_6
          & 588/9636
          & \makecell{\textbf{51}/178 \\ \negchg{-91\%} / \negchg{-98\%}}
          & \makecell{\textbf{51}/5 \\ \negchg{-91\%} / \negchg{-100\%}}
          & \makecell{210/13 \\ \negchg{-64\%} / \negchg{-100\%}}
          & \makecell{210/\textbf{0} \\ \negchg{-64\%} / \negchg{-100\%}}
          &
          & 10224/0
          & \makecell{- \\ -}
          & \makecell{\textbf{156}/0 \\ \negchg{-98\%} / \neuchg{0\%}} \\

        \bottomrule
      \end{tabular}
    }
  \end{minipage}

  % ==========================================================
  \caption{Resource usage of the parallel quantum walk (PQW) benchmark across different working-qubit scales (e.g., \texttt{PQW18\_2} denotes 18 working qubits and 2 parallel layers). The top table reports circuit width and depth. Original denotes the unoptimized max-ancilla version using dirty ancillas (left) or clean ancillas (right). For the dirty-ancilla versions, the \trivD and \bona columns correspond to the \recycle+\trivD\ and \recycle+\bona\ pipelines, respectively. The Manual+\trivD and Manual+\bona columns apply the corresponding pipelines to the manually optimized versions.
  Ancillas introduced by the decomposition of $\tMCX{n}$ gates are not manually optimized. For clean-ancilla versions, \recycle\ and \quantinuum\ reuse clean ancillas. The second line of each cell reports the relative change (\%) with respect to the corresponding dirty- or clean-ancilla Original circuit, and ``--'' indicates that the optimization timed out after 8 hours. The bottom table reports clean/dirty ancilla usage (C/D) for the same set of PQW benchmarks.}

  \vspace{-0.25cm}
  \label{tab:pqw}
\end{table*}

\subsubsection{Shor's Algorithm}

One of the best-known applications of dirty ancillas is the substantial width
reduction achieved in Shor's algorithm~\cite{factoring2np2_2017,
gidney2018factoringn2cleanqubits}. Prior work achieves this reduction using an
intricate, Shor-specific ancilla-reuse strategy. This case study evaluates
whether \bona\ can capture such algorithm-specific structure and automatically recover most of the reuse benefit.

\paragraph{Implementation details}
We implement modular exponentiation circuits, which constitute the entire quantum part of Shor's algorithm, following the construction in~\cite{factoring2np2_2017}. 
Our implementation includes both a dirty-ancilla max-ancilla version and a manually optimized version based on the intricate ancilla-reuse strategy developed in that work. For simplicity, we replace the QFTs with $\tMCX{n}$ gates, as this replacement does not affect the available reuse opportunities.
We retain prime input instances $N$ of Shor's algorithm to study how \bona\ scales with the circuit size induced by increasing problem size. In each case, we choose $a$ such that $\gcd(a,N)=1$, ensuring that modular multiplication by $a$ is reversible.

\paragraph{Experimental results}

As shown in Table~\ref{tab:shor_comparison}, manual optimization (column Manual) eliminates all dirty ancillas by coordinating reuse across the repeated components of Shor circuits. Both \trivD\ (column \trivD) and \bona\ (column \bona) substantially reduce dirty ancillas; for example, \bona\ reduces their number from 1008 to 7 for \(N=5\) and from 6400 to 20 for \(N=21\). On these Shor circuits, \trivD\ achieves lower width than \bona, as directly chaining non-overlapping dirty-qubit lifetimes matches their long sequential structure. However, \bona\ relies on local dependency and lifetime information and may therefore miss globally coordinated reuse opportunities across repeated components. Overall, \bona\ greatly reduces circuit width relative to the original circuits while maintaining depth comparable to manual optimization.

\begin{table}[t]
\centering
\scriptsize
\begin{tabular}{cccccccccccccc}
\toprule
\multicolumn{2}{c}{\textbf{Instance}}
    & \multicolumn{3}{c}{\textbf{Original}}
    & \multicolumn{3}{c}{\textbf{Manual}}
    & \multicolumn{3}{c}{\textbf{\trivD}}
    & \multicolumn{3}{c}{\textbf{\bona}} \\
\cmidrule(lr){3-5} \cmidrule(lr){6-8} \cmidrule(lr){9-11} \cmidrule(lr){12-14}
$N$ & $a$ & Width & Dirty & Depth & Width & Dirty & Depth & Width & Dirty & Depth & Width & Dirty & Depth\\
\midrule
 3 &  2 &    198 &    192 &    1484 &      6 &      0 &    1508 &     10 &      4 &    1484 &      9 &      3 &    1486 \\
 5 &  2 &   1016 &   1008 &    8038 &      8 &      0 &    8263 &     14 &      6 &    8084 &     15 &      7 &    8074 \\
 7 &  3 &   1016 &   1008 &    8176 &      8 &      0 &    8382 &     14 &      6 &    8221 &     15 &      7 &    8229 \\
 9 &  2 &   2826 &   2816 &   16809 &     10 &      0 &   17395 &     17 &      7 &   19279 &     20 &     10 &   20146 \\
15 &  2 &   2826 &   2816 &   17162 &     10 &      0 &   17805 &     17 &      7 &   19602 &     20 &     10 &   20388 \\
15 &  4 &   2826 &   2816 &   17161 &     10 &      0 &   17802 &     17 &      7 &   19634 &     20 &     10 &   20366 \\
21 &  2 &   6412 &   6400 &   43514 &     12 &      0 &   45166 &     20 &      8 &   48368 &     32 &     20 &   49156 \\
25 &  2 &   6412 &   6400 &   44286 &     12 &      0 &   46062 &     20 &      8 &   49777 &     32 &     20 &   49816 \\
\bottomrule
\end{tabular}
\caption{Comparison of Different Shor Algorithm Implementations.
\(N\) is the integer to factor and \(a\) is the randomly chosen base. For each instance we compare the original modular exponentiation circuit with manual reuse, \trivD, and \bona. We compare the circuit width, number of dirty qubits, and circuit depth.}
\vspace{-0.25cm}
\label{tab:shor_comparison}
\end{table}

\subsection{Approaching Manual Optimization for Component-level Quantum Circuits}

Components of quantum algorithms often implement focused and self-contained functionalities, are typically of moderate size, and are carefully optimized by hand. After evaluating \bona\ on quantum algorithms, we investigate how closely \bona\ approaches manually optimized component-level quantum circuits to answer \textbf{RQ1}.

\paragraph{Benchmark details}
Our benchmarks include: \texttt{MCX400}, a $\tMCX{400}$ circuit adapted from~\cite{gidney2015}, which inherently requires at least one dirty ancilla; \texttt{Incrementer50}, a 50-qubit incrementer from the same source~\cite{gidney2015}, where decomposing the $\tMCX{49}$ via an M-chain uses 47 dirty ancillas; grover15~\cite{grover}, adapted from~\cite{unqomp}, whose oracle employs an M-chain decomposition with a dirty-ancilla requirement of 13; and \texttt{StatePreparation5}~\cite{Low2024tradingtgatesdirty}, a 5-qubit state-preparation circuit, whose optimal configuration uses one clean and 80 dirty ancillas.
We use the ancilla-reuse patterns prescribed by or directly derived from the
cited constructions to construct the manually optimized versions.

\paragraph{Experimental results}

As shown in~\cref{tab:componentcircuits}, \bona\ matches the final width of the
manual implementations on \texttt{MCX400}, \texttt{Incrementer50}, and
\texttt{Grover15}, while producing a circuit that is 10 qubits wider on
\texttt{StatePreparation5}. For \texttt{MCX400}, \bona\ also produces smaller
depth than manual reuse (2788 versus 3183). On the remaining three benchmarks,
its depth is identical to that of the manual implementation or differs by at most three levels.
Thus, \bona\ automatically matches the manually designed width--depth tradeoff
on three of the four components, with a small remaining width gap on
\texttt{StatePreparation5} (\textbf{RQ1}).

\begin{table}[t]
    \centering
    \scriptsize
    \setlength{\tabcolsep}{3.5pt}
    \begin{tabular}{lccccccccc}
        \toprule
        \multirow{2}{*}{\makecell{\textbf{Circuit} \\ Dirty}} 
            & \multicolumn{3}{c}{\textbf{Original}} 
            & \multicolumn{3}{c}{\textbf{\bona}} 
            & \multicolumn{3}{c}{\textbf{Manual}} \\
        \cmidrule(lr){2-4} \cmidrule(lr){5-7} \cmidrule(lr){8-10}
            & Width& Depth & C/D 
            & Width& Depth & C/D 
            & Width& Depth & C/D \\
        \midrule
        \texttt{MCX400}~\cite{gidney2015}     
            & 1199 & 2192 & 0/797
            & \makecell{403 \\ \negchg{-66\%}} 
              & \makecell{2788 \\ \poschg{+27\%}} 
              & \makecell{0/1 \\ \textcolor{gray}{0\%} / \negchg{-100\%}}
            & \makecell{403 \\ \negchg{-66\%}} 
              & \makecell{3183 \\ \poschg{+45\%}} 
              & \makecell{0/1 \\ \textcolor{gray}{0\%} / \negchg{-100\%}} \\
        \texttt{Incrementer50}~\cite{gidney2015}       
            & 1178 & 4515 & 0/1128
            & \makecell{97 \\ \negchg{-92\%}}
              & \makecell{4515 \\ \textcolor{gray}{+0\%}}
              & \makecell{0/47 \\ \textcolor{gray}{0\%} / \negchg{-96\%}}
            & \makecell{97 \\ \negchg{-92\%}}
              & \makecell{4515 \\ \textcolor{gray}{+0\%}}
              & \makecell{0/47 \\ \textcolor{gray}{0\%} / \negchg{-96\%}} \\
        \texttt{Grover15}~\cite{grover, unqomp}    
            & 3566 & 14488 & 0/3550
            & \makecell{29 \\ \negchg{-99\%}}
              & \makecell{14488 \\ \textcolor{gray}{+0\%}}
              & \makecell{0/13 \\ \textcolor{gray}{0\%} / \negchg{-100\%}}
            & \makecell{29 \\ \negchg{-99\%}}
              & \makecell{14488 \\ \textcolor{gray}{+0\%}}
              & \makecell{0/13 \\ \textcolor{gray}{0\%} / \negchg{-100\%}} \\
        \texttt{StatePreparation5}~\cite{Low2024tradingtgatesdirty}  
            & 229 & 3276 & 2/212
            & \makecell{106 \\ \negchg{-54\%}}
              & \makecell{3497 \\ \poschg{+7\%}}
              & \makecell{1/90 \\ \negchg{-50\%} / \negchg{-58\%}}
            & \makecell{96 \\ \negchg{-58\%}}
              & \makecell{3500 \\ \poschg{+7\%}}
              & \makecell{1/80 \\ \negchg{-50\%} / \negchg{-62\%}} \\
        \bottomrule
    \end{tabular}
    \caption{
    Width, Depth, and clean/dirty ancilla counts (C/D) for component-level benchmark circuits using dirty-ancilla implementations. 
    From left to right: the max-ancilla versions (Original), \recycle+\bona\ pipeline (\bona), and the manually optimized versions (Manual).
    }
    \label{tab:componentcircuits}
    \vspace{-0.6cm}
\end{table}

\subsection{Unlocking the Scheduling Potential of Dirty Ancillas in Parallel Circuits}

\Cref{tab:pqw} primarily answers \textbf{RQ2} by showing that, for circuits with parallel execution structures, dirty-ancilla implementations may initially have greater depth, but after applying the \recycle+\bona\ pipeline, they achieve substantially lower depth than the optimized clean-ancilla implementations while attaining comparable final width. In this subsection, we further investigate this scheduling advantage to answer \textbf{RQ2}. Specifically, we randomly compose small circuits into larger circuits with parallel execution structures of varying degrees of parallelism, and examine how the advantage changes and at what level of parallelism it emerges.

\paragraph{Implementation Details}
We derive two subsets from RevLib:
(1) circuits containing no ancillas (17 circuits in total), where ancillas are introduced only through $\tMCX{n}$ decompositions before optimization; and
(2) circuits obtained by excluding extremely large instances (85 circuits in total), to avoid excessive variance during stitching.
For each subset, we generate 10 circuit sets, each containing 100 randomly selected components.
Given a parallelism parameter $p$, the 100 components in each set are stitched into layers with approximately $p$ parallel subcircuits per layer, resulting in about $100/p$ layers. For each subset and each value of $p$, this procedure yields 10 composite circuits, over which we report the mean.
Varying $p$ therefore corresponds to evaluating different degrees of parallelism.

\begin{figure}[htbp]
    \centering
    \includegraphics[width=1.0\linewidth]{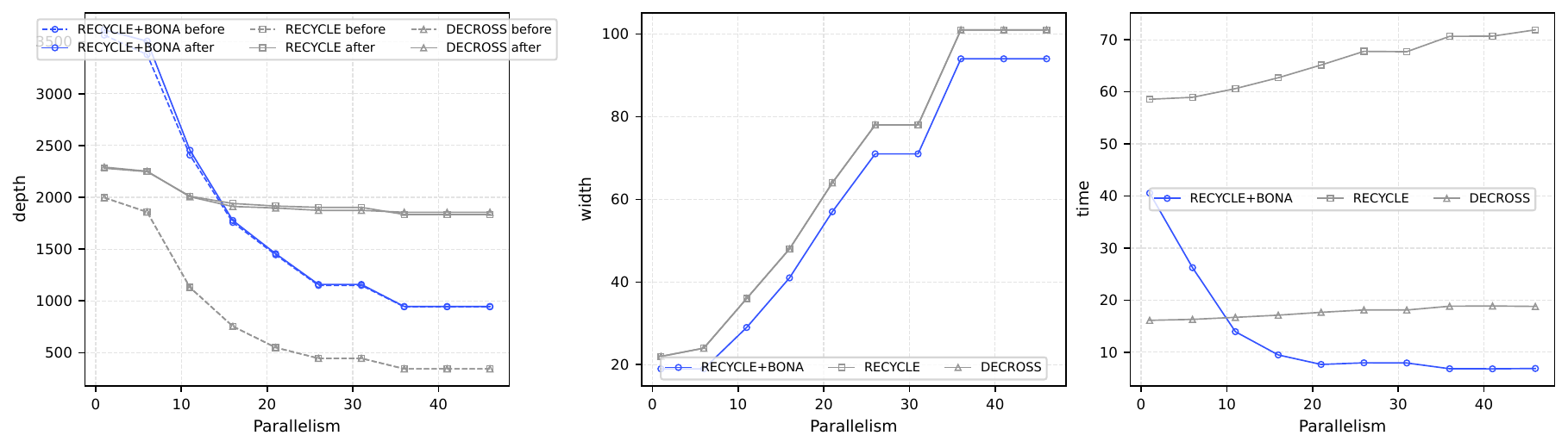}
    \vspace{0.8em}
    \includegraphics[width=1.0\linewidth]{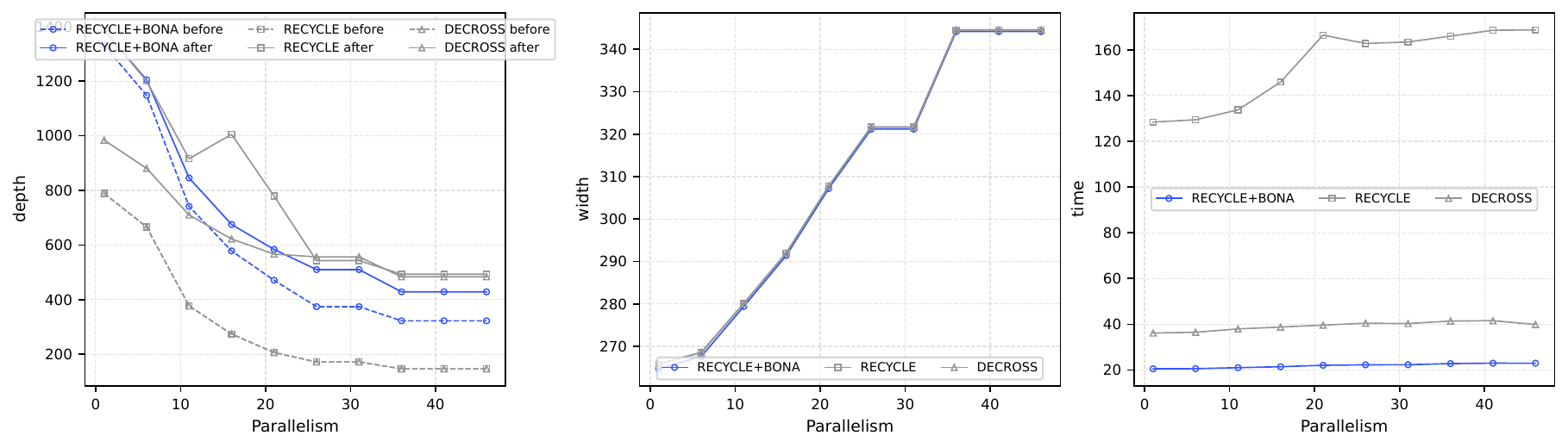}
    \vspace{-0.5cm}
    \caption{
    From left to right, the plots show circuit depth, width, and optimization time as parallelism increases for the three optimization strategies \recycle\ +\bona, \recycle, and \quantinuum; lower is better for all three metrics.
    Each data point reports the mean over 10 composite circuits.
    The top and bottom rows correspond to the two RevLib subsets, respectively.
    For depth, dashed lines indicate pre-optimization values, while solid lines indicate post-optimization results.
    For width, solid lines report post-optimization results; all three methods substantially reduce ancilla usage.
    The three strategies are distinguished by color and marker shape: blue circles for \recycle\ +\bona, grey squares for \recycle, and grey triangles for \quantinuum.
    The \recycle\ +\bona\  pipeline operates on composite circuits whose $\tMCX{n}$ gates are decomposed using dirty ancillas, whereas the standalone \recycle\ and \quantinuum\ strategies operate on clean-ancilla decompositions.}
% \vspace{-0.2cm}
    \label{fig:random-circuits}
\end{figure}

\paragraph{Experimental results}
As shown in the two leftmost plots of~\cref{fig:random-circuits}, the dirty-ancilla circuits have, on average, roughly twice the pre-optimization depth of the clean-ancilla circuits. After optimization, \recycle+\bona\ increases the depth of the dirty-ancilla versions only slightly, whereas both \recycle\ and \quantinuum\ substantially increase the depth of the clean-ancilla versions (\textbf{RQ2}).

For subset~(1), the dirty-ancilla depth becomes consistently smaller once $p$ reaches approximately 15. This corresponds to about seven layers with 15 parallel subcircuits per layer. For subset~(2), the crossover occurs at approximately $p=22$. At the highest parallelism level in the top-left plot, corresponding to approximately two layers with 50 parallel subcircuits each, the optimized clean-ancilla depth is nearly twice the dirty-ancilla depth.

The two middle plots, which report post-optimization width, further show that the dirty-ancilla versions consistently maintain smaller width in the top-middle plot and comparable width in the bottom-middle plot across all parallelism settings (\textbf{RQ2}). On average, in the top-middle plot, \recycle+\bona\ leaves only one dirty ancilla unreduced, whereas both \recycle\ and \quantinuum\ leave approximately seven clean ancillas unreduced.

In addition, the two rightmost plots show that \recycle+\bona\ requires less optimization time than \recycle\ and \quantinuum\ (\textbf{RQ3}).

\subsection{Summary and Future Work}
Overall, our evaluation shows that \bona\ is effective in reusing dirty ancillas. In practice, it can be integrated with clean-ancilla optimization, as illustrated by the \recycle+\bona\ pipeline, to handle circuits containing both types of ancillas. Our results also reveal a distinct scheduling advantage of dirty ancillas, particularly in parallel circuits. Thus, clean and dirty ancillas can be treated as complementary resources, whose respective advantages can be jointly exploited in circuit optimization.

Besides scheduling, another important problem is facilitating the use of dirty ancillas, similar to prior work on clean ancillas~\cite{silq,unqomp,reqomp,modularuncomp,qurts}. Such automation may be provided through type checking~\cite{silq,qurts} or circuit-structure analysis~\cite{unqomp,reqomp}. Since dirty ancillas are more complex to use than clean ancillas, new methods may be required.

\section{Related Work}
\label{sec: related work}

Existing work on automated ancilla management is primarily framed as a
\emph{logical circuit-level} optimization problem, rather than one tied to
hardware topology. This aligns with the layered structure of quantum
compilation: resource management is performed on the logical circuit,
while hardware-specific tasks such as mapping, routing, and SWAP insertion
are handled later. Operating at this level makes the optimization more
general and reusable across backends, as it targets intrinsic dataflow
rather than architectural constraints. Topology-aware factors are thus
largely orthogonal, affecting cost models but not the core formulation of
ancilla allocation and reuse. Following this perspective, prior work has
proposed several logical-level approaches, which we review next.

\paragraph{Clean qubit management.}
Qubit management, which relies on the reuse of clean qubits (also called recycling), has been widely studied.
Paler et al.\cite{Paler2017Wirerecycling} first introduces causal graphs to apply automatic recycling to quantum circuits. Sadeghi et al.\cite{sadeghi2022quantumcircuitresizing} and DeCross et al.\cite{quantinuum} respectively introduce dependency lists and the causal cone to reuse. The latter also uses the and dual-circuit technique.
Hua et al.~\cite{Hua2023CqQR} developed CaQR, a tool that exploits trade-offs among qubit reuse, fidelity, gate count, and circuit duration, while also handling gate commutativity.
Jiang~\cite{jiang2024recycling} introduced qubit dependency graph abstraction to design an efficient solver for the recycling problem and further proving the problem’s NP-completeness by reducing Wilf’s question to the decisional version.

\paragraph{Beyond structural topology.}
The following recent works exceed the scope of this paper and allow structural circuit modifications.
SQUARE\cite{Ding2020SQUARE} strategically rearranges uncomputation blocks to create opportunities for qubit reuse.
Furthermore, Reqomp~\cite{reqomp} reuses ancillas prior to uncomputation and supports automatic uncomputation.
The work of Fang et al.~\cite{fang2023Dynamic} introduces a general framework that enables optimal dynamic quantum circuit compilation via qubit-reuse while also managing commutable structures.
Brandhofer et al.\cite{Brandhofer2023Nearterm} combines a SAT-based model with per-qubit error characterization, demonstrating its effectiveness for circuits with up to ten qubits.
Niu et al.\cite{Niu2024OptimizationinBQSKit} introduces reuse opportunities at the unitary-matrix level and provides an algorithm configurable between qubit reduction and circuit-depth reduction. 
Tang et al.~\cite{Tang2025Widthoptimization} uncovered additional reuse opportunities by focusing on adjusting the gate execution sequence.
Kim et al.\cite{Kim2025QRmap} incorporates the mapping onto physical qubits.

\paragraph{Dirty qubits with manual borrowing}
Dirty qubits have been widely used in circuit design and quantum algorithms. 
Barenco et al.\cite{Barenco1995Elementarygates} were the first to investigate the construction of $\tMCX{n}$ gates using dirty ancillas.
Gidney presented an incrementer circuit employing dirty ancillas with a detailed explanation in his well-known blog~\cite{gidney2015}.
Häner et al.\cite{factoring2np2_2017} introduced a constant adder based on dirty ancillas for Shor's algorithm. Gidney combines the dirty-ancilla techniques in the previous two works to further reduce the number of clean qubits required in Shor’s algorithm~\cite{gidney2018factoringn2cleanqubits}.
More recently, Nie et al.\cite{nie2024quantumcircuitmultiqubittoffoli} proposed the conditional clean-qubit strategy for $\tMCX{n}$, achieving $O(\log n)$ depth and $O(n)$ size while requiring only a single ancilla of either type.
Building on this strategy, Khattar et al.\cite{Khattar2025riseofconditionally} introduced ladder-toggle detection to replace clean ancillas with dirty ones under certain conditions and demonstrates several constructions with reduced depth, gate count, and overall resource usage.
Low et al.\cite{Low2024tradingtgatesdirty} developed a data-lookup oracle in which dirty ancillas serve as batched temporal workspace, enabling a trade-off between dirty qubits and $\tT$-gate count for state preparation and unitary synthesis.
Similarly, Huang et al.\cite{Huang2025Constructingquantumimplementations} exploited dirty ancillas to minimize $\tT$-depth or width in cryptographic circuit gadgets.
More recently, Remaud et al.\cite{Remaud_2025} developed the first ancilla-free quantum adder with sublinear depth using the technique of dirty qubits.

\paragraph{Verification of safe use of dirty ancillas}

Su et al.\cite{su2024bibasedreasoningquantumprograms, su2025borrowingdirtyqubitsquantum} formalize the semantics of dirty-qubit borrowing in quantum programming languages, present quantum separation logic, and use SMT solvers to verify the safe use of dirty ancillas.
Notably, while Su et al.~\cite{su2025borrowingdirtyqubitsquantum} reduce the
\emph{verification of safe use} to the unsatisfiability of Boolean formulas,
our work establishes the NP-hardness of a fundamentally distinct optimization
problem via a reduction from 3-SAT to \fr.
Their results indicate that ensuring the safe use of dirty ancillas requires no additional reasoning beyond standard quantum program verification~\cite{Ying24,LSZreview,CVLreview}, implying that recently developed automated verification tools such as~\cite{VerifyingWithTreeAutomata25,chen2023automata} may also be capable of reasoning about dirty ancillas.
Recently, \cite{formalVeriSafety} formally verify clean and dirty ancilla safety via Pauli-X/Z commutativity checks.

\section*{Data-Availability Statement}\label{sec: data-availability}

The implementation and benchmarks used in \cref{sec:evaluation} are available at \url{https://anonymous.4open.science/r/ReuseDQ-D412}. 
The repository includes source code and experimental scripts necessary to reproduce the results.

\section*{Acknowledgements}

We thank the anonymous reviewers for their helpful feedback, which helped us improve the paper. We are also grateful to Minbo Gao, Zhenhao Li and Qisheng Wang for valuable technical discussions.
This work was supported in part by the Beijing Major Science and Technology Project under Contract no. Z251100008125035. This work was supported by Beijing Academy of Artificial Intelligence (BAAI).

\bibliographystyle{ACM-Reference-Format}
\bibliography{references}

\newpage
\appendix

\noindent{\LARGE \textbf{\textsf{Supplementary Material}}}

\section{Proof of Proposition~\cref{prop:semantics-preserving-tarnsitions}}\label{app:proof-prop-semantics-preserving-tarnsitions}

We begin by proving semantics of substituition.
\begin{proposition}\label{prop:semantics-subst}
    $\sem{C[q/d]} = \sem{C}[q/d]$.
\end{proposition}
\begin{proof}[Proof]
    We proceed by structural induction on $C$.

    \noindent\textbf{1.\;Case $C = U[\overline{p}]$.}
    This follows directly from Definition~\ref{def:semantics} and the definition of substituition:
    \[
        \sem{U[\overline{p}]}[q/d]
        = (U_{\overline{p}})[q/d]
        = U_{\overline{p[q/d]}}
        = \sem{U[\overline{p[q/d]}]}
        = \sem{\,U[\overline{p}][q/d]\,}.
    \]

    \noindent\textbf{2.\;Case $C = C_1; C_2$.}
    Using the induction hypothesis:
    \begin{align*}
        \sem{(C_1; C_2)[d/q]}
            &= \sem{(C_1[d/q]); C_2[d/q]}  \\
            &= \sem{C_2[d/q]}\, \sem{C_1[d/q]} \\
            &= \sem{C_2}[d/q]\, \sem{C_1}[d/q] \\
            &= (\sem{C_2}\sem{C_1})[d/q] \\
            &= \sem{C_1; C_2}[d/q].
    \end{align*}

    \noindent\textbf{3.\;Case $C = \borr{a}{C'}$.}
    We distinguish several subcases depending on the relationship among $a, q, d$.

    \textit{(a) If $d = a$.}
    Then substitution does not affect the borrow binder:
    \[
        \sem{\borr{a}{C'}[q/d]} = \sem{\borr{a}{C'}}
        = F_{\qv(C') \setminus a}
        = F_{\qv(C') \setminus a}[q/a],
    \]
    or both sides equal $\bot$ if the borrow is unsafe.

    \textit{(b) If $q = a$ and $d \neq a$.}
    Alpha-renaming reduces this case to the inductive hypothesis:
    \begin{align*}
        \sem{\borr{a}{C'}[a/d]}
            &= \sem{\borr{a'}{C'[a'/a]}[a/d]} \\
            &= \sem{\borr{a'}{C'[a'/a][a/d]}} \\
            &= \sem{\borr{a'}{C'[a'/a]}}[a/d].
    \end{align*}

    \textit{(c) If $q \ne a$ and $d \ne a$.}
    Assuming the borrow is safe, we compute:
    \begin{align*}
        \sem{(\borr{a}{C})[q/d]} 
        \ket{e_i}_{\overline{p}} \ket{x}_d \ket{y}_a \ket{z}_q
            &= \sem{C[q/d]} 
               \ket{e_i}_{\overline{p}} \ket{x}_d \ket{y}_a \ket{z}_q \\
            &= \sem{C}[q/d]
               \ket{e_i}_{\overline{p}} \ket{x}_d \ket{y}_a \ket{z}_q \\
            &= (F_{\overline{p},d} \otimes I_a)[q/d]
               \ket{e_i}_{\overline{p}} \ket{x}_d \ket{y}_a \ket{z}_q \\
            &= F_{\overline{p},q} \otimes I_a
               \ket{e_i}_{\overline{p}} \ket{x}_d \ket{y}_a \ket{z}_q.
    \end{align*}
    Hence,
    \[
        \sem{(\borr{a}{C})[q/d]}
        = F_{\overline{p},q}
        = F_{\overline{p},d}[q/d]
        = \sem{\borr{a}{C}}[q/d].
    \]
    If the borrow is unsafe, both sides evaluate to $\bot$, so the equality is trivial.
\end{proof}

\begin{proposition}\label{prop:semantics-transition}
    All transition rules in~\cref{subsec:transitions}, except \RuleSubst, are semantics preserving. 
    The rule \RuleSubst~ is semantics preserving provided that the statement $\borr{d}{C}$ is safe.
\end{proposition}

\begin{proof}
    We prove semantic preservation by considering the three classes of transition rules.

    \paragraph{Topological transitions.}
The cases \RuleSwap~ and \RuleAssoc~ are immediate, as they merely reorder or regroup sequential compositions without affecting denotational semantics.

    We next consider \RuleBorrLeft.  
    Assume $d \notin \qv(C_0)$.  
    Let $\overline{p}$ denote the qubits in $(\qv(C_1)\setminus d)\setminus \qv(C_0)$ and $\overline{q}$ denote the qubits in $\qv(C_0)$.

    If either side is unsafe, then both  
    $\sem{C_0; \borr{d}{C_1}}$ and $\sem{\borr{d}{C_0; C_1}}$ equal $\bot$, and the equality is trivial.  
    Suppose now that both sides are safe.  
    For any state $\ket{x}_d$, computational basis states $\ket{e_i}_{\overline{p}}$ and $\ket{f_j}_{\overline{q}}$, we have:
    \begin{align*}
        \sem{C_0; \borr{d}{C_1}} 
        \ket{x}_d \ket{e_i}_{\overline{p}} \ket{f_j}_{\overline{q}}
        &= \sem{\borr{d}{C_1}}
           \ket{x}_d \ket{e_i}_{\overline{p}} \sem{C_0}\ket{f_j}_{\overline{q}} \\
        &= \sem{C_1} 
           (\ket{x}_d \ket{e_i}_{\overline{p}} \sem{C_0}\ket{f_j}_{\overline{q}}) \\
        &= \ket{x}_d F (\ket{e_i}_{\overline{p}} \sem{C_0}\ket{f_j}_{\overline{q}})
    \end{align*}
    where $\sem{C_1} = F \otimes I_d$ because $\sem{\borr{d}{C_1}}$ is safe.

    Similarly,
    \begin{align*}
        \sem{\borr{d}{C_0; C_1}}
        &= \sem{C_0; C_1}\ket{x}_d \ket{e_i}_{\overline{p}} \ket{f_j}_{\overline{q}} \\
        &= \sem{C_1}
           (\ket{x}_d \ket{e_i}_{\overline{p}} \sem{C_0}\ket{f_j}_{\overline{q}}) \\
        &= \ket{x}_d F (\ket{e_i}_{\overline{p}} \sem{C_0}\ket{f_j}_{\overline{q}}),
    \end{align*}
    which matches the expression above.  
    Hence $\sem{C_0; \borr{d}{C_1}} = \sem{\borr{d}{C_0; C_1}}$.

    The rule \RuleBorrRight~ is analogous.

    We finally consider \RuleBorrComm.  Let
    $R = \qv(C) \setminus \{d_1,d_2\}$ and $U = \sem{C}$.  Expanding the
    semantics of the two nested borrowing statements shows that
    \[
        \sem{\borr{d_1}{\borr{d_2}{C}}} = F_R
    \]
    if and only if
    \[
        U = F_R \otimes I_{\{d_1,d_2\}}.
    \]
    The same condition, with the same result $F_R$, characterizes
    $\sem{\borr{d_2}{\borr{d_1}{C}}}$.  If this factorization does not
    exist (including when $U=\bot$), both sides evaluate to $\bot$.
    Therefore, exchanging the order of two distinct nested borrowing
    statements preserves both safety and denotational semantics.

    \paragraph{Structural transitions.}
    The rules \RuleCompLeft, \RuleCompRight, and \RuleBorrInner~ are semantics preserving directly by induction on the sub-derivations
    
    \paragraph{\RuleSubst.}
    Finally, consider substitution.  
    To show that \RuleSubst~ is semantics preserving, it suffices to prove  
    \[
        \sem{C[q/d]} = \sem{C}[q/d],
    \]
    under the assumptions that $q \notin \qv(C) \cup \qa(C)$ and that $\borr{d}{C}$ is safe.
    This is precisely Proposition~\ref{prop:semantics-subst}.
\end{proof}

\begin{proposition}
    For any safe circuit $C$ such that $C \to^{*} C'$, we have
    \[
    \sem{C'} = \sem{C}.
    \]
\end{proposition}
\begin{proof}
    This is trivial with Proposition~\ref{prop:semantics-transition}
\end{proof}

\section{Proofs for the Endpoint DAG Model}\label{app:correctDAG}

In this appendix, we provide the full proofs for the properties of the
Endpoint DAG model introduced in \Cref{sec:heuristic:graph}.

\subsection{Proof of Lemma \ref{lem:dag-properties} (Validity and Invariance)}

\begin{proof}
The lemma claims two fundamental properties: \textbf{Validity} (a DAG is
valid if and only if it corresponds to a circuit) and \textbf{Invariance}
(topological/structural rules do not alter the DAG). We prove them
separately.

\paragraph{Proof of Validity.}
($\Rightarrow$) Suppose $G = \T(C)$ is constructed from a valid quantum
circuit $C$.
(1) \emph{Acyclic}: Quantum circuits are inherently sequential and
operations strictly follow the forward progression of time. Because directed
edges are only added from earlier operations to later operations within the
abstract syntax tree, no cycles can be formed.
(2) \emph{Balanced}: Every quantum gate $g$ acting on $k$ qubits takes
exactly $k$ input wires and produces exactly $k$ output wires. Thus, its
corresponding node $\Op(g)$ will have exactly $k$ incoming edges and $k$
outgoing edges, yielding $\deg^{-}(\Op(g)) = \deg^{+}(\Op(g))=k$.
(3) \emph{Input/Output Consistency}: For every qubit $q$ (whether working or
dirty), the construction creates a distinguished path $P_q$ from $\In(q)$ to
$\Out(q)$, and the edge multiset is exactly the multiset-disjoint union of
these qubit paths.

($\Leftarrow$) Suppose $G = (Q, Q_D, V, E)$ satisfies the three validity
conditions. Since $G$ is acyclic, it admits a topological ordering of its
vertices. We proceed by induction on the number of operation nodes in $V$.
\begin{itemize}
    \item \emph{Base Case}: If there are no operation nodes, the balanced
and I/O consistency conditions imply that for every $q \in Q$, there is a
direct edge $(\In(q), \Out(q))$. This corresponds to an empty circuit
(identity operation).
    \item \emph{Inductive Step}: Assume any valid DAG with $k$ operation
nodes corresponds to a circuit. Consider a DAG with $k+1$ operation nodes.
Fix the qubit-wise path decomposition $\{P_q\}_{q\in Q}$ guaranteed by
I/O consistency.
In the topological sort of $G$, let $u = \Op(g)$ be the first operation
node. Because it is the first operation, all its incoming edges must
originate from input nodes, say $\In(q_1), \dots, \In(q_\ell)$.
Let $v_i$ be the successor of $u$ on the path $P_{q_i}$, i.e., $P_{q_i}$ locally contains
$\In(q_i) \to u \to v_i$.
    We construct a new graph $G'$ by removing $u$. For each $i \in \{1,
\dots, \ell\}$, we delete the edges $(\In(q_i), u)$ and $(u, v_i)$, and add
a direct edge $(\In(q_i), v_i)$. This operation preserves acyclicity,
balance, and I/O consistency, resulting in a valid DAG $G'$ with $k$
operation nodes. By the induction hypothesis, $G' = \T(C')$ for some circuit
$C'$. We then construct $C = g[q_1, \dots, q_\ell] ; C'$. It is
straightforward to verify that $\T(C) = G$.
\end{itemize}

\paragraph{Proof of Invariance.}
We must show that if $C_1 \topotrans^{*} C_2$, then
$\T(C_1) = \T(C_2)$. It suffices to consider each generating topological
rule; applications under structural contexts do not affect the argument.
\begin{itemize}
    \item \RuleSwap: $C_A ; C_B \leftrightarrow C_B ; C_A$ where $C_A$ and
$C_B$ act on disjoint sets of qubits. In the construction of $\T(C)$, edges
are exclusively added between sequential operations on the \emph{same}
qubit. Since $C_A$ and $C_B$ share no qubits, the path of any qubit $q$
passes through either $C_A$, $C_B$, or neither, but never both. Therefore,
changing their relative textual order in the syntax does not alter the
predecessor-successor relationship of any gate on any specific qubit wire.
Thus, the edge set $E$ remains identical.
    \item \RuleAssoc: Sequence concatenation is strictly associative; the
sequence of gates applied to any qubit $q$ is identical in both
$(C_A;C_B);C_C$ and $C_A;(C_B;C_C)$.
    \item \RuleBorrLeft\ / \RuleBorrRight:
    $\borr{d}{C_0}; C_1 \leftrightarrow \borr{d}{C_0; C_1}$
    given $d \notin \qv(C_1)$. The translation $\T$ does not
encode the boundaries of the \texttt{borrow} block as explicit nodes; it
only tracks the operational flow. Expanding the lexical scope of $d$ over
$C_1$ does not insert any new operations on $d$ (since $d \notin \qv(C_1)$),
nor does it change the relative sequence of operations on any other qubits.
Thus, the constructed DAG is invariant. The \RuleBorrLeft\ case is symmetric.
    \item \RuleBorrComm:
    Exchanging two nested borrowing declarations changes neither $Q$, $Q_D$,
    nor the gate sequence on any qubit, so the DAG remains unchanged.
\end{itemize}
The result for $\topotrans^{*}$ follows by induction on the transition-sequence
length. This concludes the proof of Lemma 4.1.
\end{proof}

%%%%%%%%%%%%%%%%%%%%%%%%%%%%%%%%%%%%%%%%%%%%%%%%%%%%%
\subsection{Proof of Unreachable Condition}

As discussed in \Cref{sec:formalization}, when splicing the path of $d$
into an edge $e=(u,v)$, the necessity of the unreachable condition ($\fst(d)
\not\leadsto u$ and $v \not\leadsto \lst(d)$) to prevent immediate cycles is
straightforward. We now formally prove that this condition is sufficient to
guarantee that the resulting graph $G[e \triangleleft d]$ satisfies all
three validity conditions of an Endpoint DAG.

\begin{lemma}[Validity of Edge-Splicing]\label{lem:splicing-validity}
Given a valid Endpoint DAG $G = (Q, Q_D, V, E)$, an edge $e=(u,v) \in E$,
and an ancilla $d \in Q_D$, if $\fst(d) \not\leadsto u$ and $v \not\leadsto
\lst(d)$ in $G$, then the spliced graph $G' = G[e \triangleleft d]$
satisfies the Acyclic, Balanced, and I/O Consistency conditions.
\end{lemma}

\begin{proof}
We verify the three validity conditions for $G'$ sequentially.

\paragraph{1. Balanced Condition.}
The splicing operation $G[e \triangleleft d]$ removes the boundary nodes
$\In(d)$ and $\Out(d)$, and the edges $e=(u,v)$, $(\In(d), \fst(d))$, and
$(\lst(d), \Out(d))$. It then adds two new edges: $(u, \fst(d))$ and
$(\lst(d), v)$.
For any operation node $w \in V' = V \setminus \{\In(d), \Out(d)\}$, its
degree changes only if it is one of the endpoints involved in the splice:
$u$, $v$, $\fst(d)$, or $\lst(d)$.
\begin{itemize}
    \item For node $u$: one outgoing edge $(u,v)$ is removed, and one
outgoing edge $(u, \fst(d))$ is added. Thus, $\deg^{+}(u)$ remains
unchanged.
    \item For node $v$: one incoming edge $(u,v)$ is removed, and one
incoming edge $(\lst(d), v)$ is added. Thus, $\deg^{-}(v)$ remains
unchanged.
    \item For node $\fst(d)$: the incoming edge from $\In(d)$ is replaced by
the incoming edge from $u$, preserving $\deg^{-}(\fst(d))$.
    \item For node $\lst(d)$: the outgoing edge to $\Out(d)$ is replaced by
the outgoing edge to $v$, preserving $\deg^{+}(\lst(d))$.
\end{itemize}
Since all other operation nodes are completely unaffected, the equation
$\deg^{-}(w) = \deg^{+}(w)$ holds for every operation node $w \in V'$.
Therefore, $G'$ remains perfectly balanced.

\paragraph{2. I/O Consistency.}
Because $d$ and its boundary nodes, $\In(d)$ and $\Out(d)$, are entirely
removed from the updated sets $Q'$ and $V'$, every remaining qubit $q \in Q'$
still retains its respective $\In$ and $\Out$ nodes in $V'$. Furthermore,
as shown in the proof balanced condition, no degree of the nodes changes.
Thus, the in-degree of each $\In(q)$ and the out-degree of each $\Out(q)$ remain exactly zero.

Although the direct edge $(u,v)$ is removed, the newly added edges
$(u,\fst(d))$ and $(\lst(d),v)$ redirect the flow through the operational
sequence of $d$. Let $\{P_q\}_{q\in Q}$ be the distinguished paths of $G$,
and suppose $(u,v)$ belongs to $P_q$. The paths $P_q$ and $P_d$ share no
operation node: any common node lies on one side of $(u,v)$ along $P_q$ and
would imply $\fst(d)\leadsto u$ or $v\leadsto\lst(d)$. We replace the edge
$(u,v)$ in $P_q$ with the segment $u \to \fst(d) \leadsto \lst(d) \to v$,
where the middle part is inherited from $P_d$, and then remove $P_d$ from the
decomposition. All other paths remain unchanged. Thus, the edge multiset $E'$ is still the
multiset-disjoint union of the distinguished paths for qubits in $Q'$.

\paragraph{3. Acyclicity.}
Assume for the sake of contradiction that the unreachable condition holds,
but $G'$ contains a cycle $C_{cycle}$. Because the original graph $G$ is
acyclic, $C_{cycle}$ must utilize at least one of the newly added edges in
$G'$: $(u, \fst(d))$ or $(\lst(d), v)$. We analyze the three possible cases:
\begin{itemize}
    \item \textbf{Case A ($C_{cycle}$ uses only $(u, \fst(d))$):} This
implies there exists a path from $\fst(d)$ back to $u$ in $G'$ that does not
use $(\lst(d), v)$. Since all other edges in this path belong to the
original graph $G$, it implies $\fst(d) \leadsto u$ in $G$. This directly
contradicts the assumption $\fst(d) \not\leadsto u$.

    \item \textbf{Case B ($C_{cycle}$ uses only $(\lst(d), v)$):} Similarly,
this implies a path from $v$ back to $\lst(d)$ entirely composed of original
edges. Thus $v \leadsto \lst(d)$ in $G$, contradicting the assumption $v
\not\leadsto \lst(d)$.

    \item \textbf{Case C ($C_{cycle}$ uses both $(u, \fst(d))$ and
$(\lst(d), v)$):} This implies that in the original graph $G$, there exists
a path from $v$ to $u$ (to close the loop from the endpoint of the second
new edge back to the start of the first). However, $G$ is a valid DAG and
already contains the forward edge $(u,v)$. If there were a path $v \leadsto
u$ in $G$, $G$ itself would contain a cycle $(u \to v \leadsto u)$, which
contradicts the foundational premise that $G$ is acyclic.
\end{itemize}
Since all three cases lead to a contradiction, no such cycle $C_{cycle}$ can
exist. Therefore, $G'$ is strictly acyclic.

Having satisfied all three validity conditions, the spliced graph $G'$ is
formally established as a valid Endpoint DAG.

\end{proof}

%%%%%%%%%%%%%%%%%%%%%%%%%%%%%%%%%%%%%%%%%%%%%%%%%%%%%%%%%%%%%%%%%%%%%%%

\subsection{Proof of Theorem~\ref{thm:step-correspondence} (Step Correspondence)}

\begin{proof} 
The theorem asserts the isomorphism between graph-level edge-splicing and
circuit-level borrowing transitions.

\paragraph{Soundness.}
Suppose $G \to G'$ via a valid edge-splicing step $G' = G[e \triangleleft
d]$ on edge $e=(u,v)$. By Lemma~\ref{lem:splicing-validity}, the
unreachable condition guarantees that $G'$ is a valid, acyclic endpoint DAG.
Therefore, $G'$ admits a valid topological ordering.

Because the edge $e$ existed in $G$, let $q$ be the specific qubit line to
which $e$ belonged in the original circuit $C$, and let $P_q$ be its
distinguished path. Using \RuleBorrLeft~and \RuleBorrRight, we first move all
borrowing declarations to an outer prefix; lexical scoping and freshness
ensure the side conditions. We then use \RuleBorrComm~to place $d$ inside all
other declarations, so the declaration of $q$, if $q$ is dirty, remains
outside that of $d$.

Fix such a topological ordering and restrict it to operation nodes. Since the
splice preserves the relative order of all gates that shared a qubit in $C$,
this ordering can be obtained from the original gate order by repeatedly
swapping adjacent disjoint gates, using \RuleAssoc~to expose them. Let $C_d$
be the consecutive code block from $\fst(d)$ through $\lst(d)$ in this
ordering, possibly containing unrelated gates. Operations on $P_q$ before
$e$ precede $\fst(d)$, those after $e$ follow $\lst(d)$, and the unreachable
condition excludes a common operation of $P_q$ and $P_d$. The prefix and
suffix outside $C_d$ contain no operation on $d$, so reverse applications of
\RuleBorrLeft~and \RuleBorrRight~shrink the scope of $d$ to $C_d$. Therefore,
these topological transitions construct a configuration $C_0$ with
$C\topotrans^{*}C_0$. By Lemma~\ref{lem:dag-properties},
$\T(C_0)=\T(C)=G$.

In $C_0$, the qubit $q$ is strictly idle during the code block $C_d$ of $d$.
Since all other borrowing declarations remain outside that of $d$, $q$ is not
declared in $C_d$ either. Hence $q$ is unrelated to the code block of $d$, so
the side-condition for the borrowing rule, $q \notin \qv(C_d) \cup \qa(C_d)$,
is perfectly satisfied. Thus, we can legally fire the \RuleSubst\ transition to
substitute $d$ with $q$, yielding a target circuit $C'$. Since the syntactic
substitution exactly merges $d$'s operations into $q$'s idle period between
$u$ and $v$, removes the boundary nodes and terminal edges of $P_d$, and
leaves every other distinguished path unchanged, it directly mirrors the
edge-splicing construction, confirming
$\T(C') = G'$.

\paragraph{Completeness.}
Conversely, suppose $C \to C'$ via a valid borrowing transition
$\RuleSubst$. In the syntax tree, the transition substitutes an ancilla $d$
with an idle qubit $q$. The validity of the transition dictates that $q
\notin \qv(C_{body}) \cup \qa(C_{body})$, meaning $q$ is neither used nor
declared during the entire scope of $d$.

In the corresponding DAG $\T(C)$, this uninterrupted idle period on $q$
naturally manifests as a single directed edge $e=(u,v)$ connecting the last
operation on $q$ before the block to the first operation on $q$ after the
block. The substitution in the redex body re-routes the causal flow of $q$
to pass entirely through the operations of $d$. By the definition of the
translation $\T$, mapping these variables directly translates to cutting the
edge $e=(u,v)$ and reconnecting $u$ to $\fst(d)$ and $\lst(d)$ to $v$, while
discarding the now-obsolete boundary nodes $\In(d)$ and $\Out(d)$. Hence
\[
    \T(C')=\T(C)[e\triangleleft d].
\]
Neither $\fst(d)\leadsto u$ nor $v\leadsto\lst(d)$ can hold in $\T(C)$:
either path would survive the splice because it cannot use a deleted terminal
edge or $e$, and together with the corresponding new edge would form a cycle
in $\T(C')$. Therefore the unreachable condition holds, and the algebraic
definition gives the edge-splicing step $\T(C)\to\T(C')$.
\end{proof}

\section{NP hardness of \fr}\label{app:NP}

In this section, we present the full proof of Lemma \ref{lem:NPdagudg} and Theorem \ref{thm:NPfr}.

\begin{proof}[Proof of Lemma \ref{lem:NPdagudg}]
   Consider a 3-SAT instance $\psi$ consisting of $m$ clauses $c_1,\dots,c_m$ over $n$ variables $x_1,\dots,x_n$. Now we construct a corresponding instance of \daginudg, which is a UDG $\mathcal{G}=(V,E,E_s,E_t)$. 
   
   For each variable $x_i$, we add $x_i$ and $\bar{x_i}$ to $V$. For each clause $c_j$, we also add $c_j$ to $V$. Furthermore, we add 2 specific vertices $c$ and $x_{false}$ to $V$. That is, $V=\{x_i\mid 1\le i\le n\}\cup\{\bar{x_i}\mid 1\le i\le n\}\cup\{c_j\mid 1\le j\le m\}\cup\{c,\xf\}$.

   We let $E=\{\xf c\}\cup\{cc_j\mid 1\le j\le m\}$. We let $Y_j$ to be the set of literals that $c_j$ contains and $E_s=\{c_jY_j\mid  1\le j\le m\}$. We let $X_i=\{x_i,\bar{x_i}\}$ and $E_t=\{X_i\xf\mid 1\le i\le n\}$. See Figure \ref{fig:3SAT} for an example.

   Now we prove that the 3-SAT instance has a satisfying assignment if and only if a DAG can be realized by the UDG we constructed. The insight of the proof is that, the realization of edges in $E_t$ corresponds to a variable assignment in the 3CNF formula by indicating either $x$ or $\bar{x}$ is false. If the assignment is a solution, then the remaining part of the UDG can be realized as a DAG (Figure \ref{fig:3SATT}). Otherwise, the violated clause will grantee a loop in such realizations (Figure \ref{fig:3SATF}).
   
   If the 3-SAT instance has a satisfying assignment, then in each clause $c_j$ there exists at least one literal $y_j$ that is true under the assignment; and we then replace the uncertain edge $c_jY_j$ with $c_jy_j$. Also, in the satisfying assignment, exactly 1 literal $z_i$ in $X_i=\{x_i,\bar{x_i}\}$ is false, and we replace the uncertain edge $X_i\xf$ with $z_i\xf$. It can then be verified that the resulting graph is a DAG with 5 levels, where every directed edge starting at a vertex in level $i$ ends at a vertex in level $i+1$.
   \begin{itemize}
       \item Level 1: Vertices corresponding to the false literals under the satisfying assignment;
       \item Level 2: $\xf$;
       \item Level 3: $c$;
       \item Level 4: Vertices corresponding to the clauses.
       \item Level 5: Vertices corresponding to the true literals under the satisfying assignment;
   \end{itemize}

   If the constructed UDG can realize a DAG, then in the DAG each uncertain edge $X_i\xf$ is replaced by $z_i\xf$ for some $z_i\in X_i=\{x_i,\bar{x_i}\}$. We claim that $\overline{z_1\dots z_n}$ is a satisfying assignment of the 3-SAT instance. Suppose otherwise, then $\overline{z_1\dots z_n}$ must violate some clause $c_j$, which means each literal in $c_j$ is false. Then, for every literal $y_j$ in $c_j$, $y_j\xf cc_j$ forms a path, and consequently there would always be a loop when replacing $c_jY_j$, a contradiction to DAG. 

   Consequently, as 3-SAT is NP-hard, \daginudg\ is NP-hard as well.
\end{proof}

\begin{proof}[Proof of Theorem \ref{thm:NPfr}]
 First, \fr\ is in NP: a certificate contains at most $|Q_D|$ edge-splicing steps, each of which can be checked in polynomial time.
 We reduce \daginudg\ to \fr. For an UDG $\mathcal{G}=(V,E,E_s,E_t)$, suppose $|V|=n$ and $|E\cup E_s\cup E_t|=m$. We now construct an endpoint DAG model $T= (Q, Q_D, U, E)$ based on $\mathcal{G}$. We fix an arbitrary order of edges in $E\cup E_s\cup E_t$, yielding $e_1,\dots,e_m$. Let $U=V\cup E\cup E_s\cup E_t\cup \{s,t,e_0\}$, which contains a corresponding element for each vertex and each edge in UDG, as well as 3 extra elements. 

 In the endpoint DAG model model $T$, $e_0, e_1,\dots,e_m\in U$ are elements used to ensure that each dirty ancilla qubit can only splice into specific positions. We use $e_{[i, j]},0\le i\le j\le m$ to denote the path $e_ie_{i+1}\dots e_j$.
 We construct the flow of $Q$ and $Q_D$, denoted by paths $P$ and $P_D$ respectively, as follows:
    \begin{itemize}
        \item For each vertex $v\in V$, create $p_v=e_{[0,m]}svt\in P-P_D$. For each edge $e_i=ab\in E$, create a path $p_{i,1}=e_{[0,i-1]}b\in P-P_D$ and a dirty path $d_{i,1}=e_{[i,m]}a\in P_D$. 
        \item For each uncertain edge $e_i=aB,B=\{b_1,\dots,b_k\}$, create $k$ paths $p_{i,j}=e_{[0,i-1]}{b_j}\in P-P_D, 1\le j\le k$,  $k-1$ dirty paths $d_{i,j}=e_{[i,m]}s\in P_D, 1\le j\le k-1$ and a dirty path $d_{i,k}=e_{[i,m]}a\in P_D$.
        \item For each uncertain edge $e_i=Ab,A=\{a_1,\dots,a_k\}$, create $k-1$ paths $p_{i,j}=e_{[0,i-1]}t\in P-P_D, 1\le j\le k-1$, a path $p_{i,k}=e_{[0,i-1]}b\in P-P_D$,  and $k$ dirty paths $d_{i,j}=e_{[i,m]}{a_j}\in P_D, 1\le j\le k$.
    \end{itemize} 
    
    Now we show that $\mathcal{G}$ can realize a DAG if and only if the dirty ancillas can be fully reused. Suppose $\mathcal{G}$ can realize a DAG $G$. Then,
    \begin{itemize}
        \item For each certain edge $e_i=ab$, let the dirty path $d_{i,1}=e_{[i,m]}a\in P_D$ splice into the path $p_{i,1}=e_{[0,{i-1}]}b$,  resulting in $p'_{i,1}=e_{[0,m]}ab\in P-P_D$.
        \item For each uncertain edge $e_i=aB,B=\{b_1,\dots,b_k\}$, suppose it is replaced by $ab_l$ in the DAG $G$. Then,
        \begin{itemize}
            \item let the dirty path $d_{i,k}=e_{[i,m]}a\in P_D$ splice into the path $p_{i,l}=e_{[0,i-1]}{b_l}$, resulting in $p'_{i,l}=e_{[0,m]}ab_l\in P-P_D$;
            \item let each dirty path of the form $d_{i,j}=e_{[i,m]}s\in P_D, 1\le j\le k-1$ splice into a path of the form $p_{i,j}=e_{[0,i-1]}{b_j}, 1\le j\le k, j\neq l$, resulting in $k-1$ paths of the form $p'_{i,j}=e_{[0,m]}sb_j\in P-P_D, 1\le j\le k, j\neq l$.
        \end{itemize}  
        \item For each uncertain edge $e_i=Ab,A=\{a_1,\dots,a_k\}$, suppose it is replaced by $a_lb$ in the DAG $G$. Then,
        \begin{itemize}
            \item let the dirty path $d_{i,l}=e_{[i,m]}a_l\in P_D$ splice into the path $p_{i,k}=e_{[0,i-1]}{b}$, resulting in $p'_{i,l}=e_{[0,m]}a_lb\in P-P_D$;
            \item let each dirty path of the form $d_{i,j}=e_{[i,m]}{a_j}\in P_D, 1\le j\le k, j\neq l$ splice into a path of the form $p_{i,j}=e_{[0,i-1]}t\in P-P_D, 1\le j\le k-1$, resulting in $k-1$ paths of the form $p'_{i,j}=e_{[0,m]}a_jt\in P-P_D, 1\le j\le k, j\neq l$.
        \end{itemize}
    \end{itemize}  

    Now we show that the resulting endpoint DAG model $T'$ is valid. In this model, every path is of one of the following forms:
    \begin{enumerate}
        \item $e_{[0,m]}sat,a\in V$;
        \item $e_{[0,m]}sb,b\in V$;
        \item $e_{[0,m]}at,a\in V$;
        \item $e_{[0,m]}ab, a,b\in V$.
   \end{enumerate}
   In particular, as $G$ is a DAG, there is a corresponding topological ordering $\tau$ of all vertices in $V$. Consequently, it can be verified that $e_0e_1\dots e_ms\tau t$ is a topological ordering of $T'$, implying $T'$ is valid.

    Now suppose the dirty ancillas can be fully reused, resulting in an endpoint DAG model $T'$. We first prove the following claim: Dirty paths in $T$ having the form $e_{[i,m]}x$ can only splice into paths having the form $e_{[0,i-1]}y$, and the resulting path is $e_{[0,m]}xy$. 
    
    The claim can be proved by induction. Suppose dirty paths in $T$ having the form $e_{[j,m]}x, 1\le j\le i-1$ can only splice into paths having the form $e_{[0,j-1]}y$, and the resulting path is $e_{[0,m]}xy$. 
    
    Consider dirty paths having the form $e_{[i,m]}x$. As $e_i$ appears in the ancilla, it can not be reused by any paths containing $e_i$. Since all paths having the form $e_{[0,j-1]}y, 1\le j\le i-1$ have been spliced, such dirty paths can only be reused by paths of the form $e_{[0,i-1]}y$. Furthermore, this splice operation can only happen between $e_{i-1}$ and $y$ since $e_i$ must be placed after $e_{i-1}$ and $e_m$ must be placed before $y$ due to the operator sequence $e_{[0,m]}syt$ \footnote{If $y=s$ or $y=t$, the statement still holds by taking the operator sequence $H_0H_1\dots H_mO_sO_xO_t,x\in V$}. This finishes the proof of the claim.

    Now we use $\mathcal{G}$ to realize a directed graph $G$. For each uncertain edge $e_i=aB,B=\{b_1,\dots,b_k\}$, suppose its corresponding paths after the splice operations are composed of a path $e_{[0,m]}a{b_l}$ for some $1\le l\le k$ and $k-1$ paths $e_{[0,m]}s{b_j},j\neq l$. Then we replace $aB$ with $ab_l$. 
    Similarly for each uncertain edge, $e_i=Ab,A=\{a_1,\dots,a_k\}$, suppose its corresponding working qubits after the reusing are composed of a qubit with operator sequence $e_{[0,m]}{a_l}{b}$ for some $1\le l\le k$ and $k-1$ working qubits $e_{[0,m]}{a_j}{t},j\neq l$. Then we replace $Ab$ with $a_lb$. 

    As $T'$ is valid, there is a topological ordering of elements in $U$, which contains a topological ordering of vertices in $G$. This implies the realized graph $G$ is a DAG. This finishes the proof.
\end{proof}

\section{Detailed Implementation of the Scheduling
Algorithm}\label{app:bona-details}

While Algorithm~\ref{alg:bona} in the main text presents the graph-theoretic
abstraction, our Python implementation utilizes a specialized structure
called \emph{OpGrid} to manage the spatiotemporal relationships between
operations efficiently.

\subsection{The OpCell and OpGrid Structure}

The \emph{OpGrid} is a sparse 2D map: $\texttt{(qubit, column)} \mapsto
\texttt{OpCell}$. Each \texttt{OpCell} represents a DAG node $v \in V$
located at a specific qubit wire and depth. It contains the following
member variables:
\begin{itemize}
    \item \textbf{node}: The unique identifier of the operation node $v$ in
the Endpoint DAG.
    \item \textbf{left} / \textbf{right}: Column indices of the immediately
preceding and succeeding \texttt{OpCell}s on the same qubit wire. These
pointers form a doubly-linked list for each wire, allowing $O(1)$ splicing.
    \item \textbf{tag}: An integer used to determine the traversal priority.
Dirty-ancilla starting nodes are assigned negative tags to ensure they are
processed before other operations in the same depth layer.
\end{itemize}

An \emph{idle edge} $(u,v)$ on a qubit $q$ is thus explicitly represented in
the OpGrid by the gap between a cell at $(q, c_u)$ and its
\texttt{cell.right} at $(q, c_v)$.

\subsection{Uniform Handling of Borrowing and Recycling}

To unify different resource management mechanisms, \bona\ strategically maps
the boundary nodes ($\In(q)$ and $\Out(q)$) from the Endpoint DAG into the
OpGrid:

\begin{itemize}
    \item \textbf{Exposing Working/Clean Wires:} For working qubits and
promoted clean ancillas, we include their $\In(q)$ nodes at column $0$ and
their $\Out(q)$ nodes at column $\text{MaxDepth}+1$. This exposes the long
idle edges at the beginning and end of the circuit to the resource pool.
Consequently, \emph{Clean-Dirty (CD) recycling} is naturally handled as a
borrowing operation into these terminal edges.
    \item \textbf{Restricting Dirty Ancilla Chaining:} Conversely, we do not
expose the boundary nodes of dirty ancillas to the resource pool during the
initial sweep. This ensures that a dirty ancilla cannot provide a boundary
candidate edge to another dirty ancilla, which would result in
premature \emph{Dirty-Dirty (DD) recycling}. Since mapping a long-lifetime
dirty qubit is strictly harder than mapping several short ones, \bona\
postpones DD-recycling until all embedding (borrowing) opportunities for
working and clean wires have been exhausted.
\end{itemize}

\subsection{Implementation Optimizations}

Algorithm~\ref{alg:detailed_bona} details the execution flow. Key
performance optimizations include:
\begin{enumerate}
    \item \textbf{Transactional Pool Updates:} The \texttt{delay\_in} buffer
(Lines 15 and 20) ensures that the Resource Pool is updated only after all
borrowing requests at the current depth are processed, preventing an ancilla
from borrowing its own internal operations (self-borrowing).
    \item \textbf{Logarithmic Tracking:} The $\Pool$ is implemented via a
\texttt{SortedList}~\cite{sortedcontainers_sortedlist} ordered by right
endpoints, ensuring $O(\log |Q|)$ query and update times.
    \item \textbf{Constant-time Acyclicity Checks:} Before scanning
candidates for hard borrowing, \bona\ pre-computes the reachable sets of the
dirty ancilla into Python \texttt{set}s ($Fwd$ and $Bwd$, Line 26--27), reducing
the graph acyclicity check to two $O(1)$ lookups (Line 30).
\end{enumerate}

\begin{algorithm}[h]
\SetKwInOut{Input}{Input}
\SetKwInOut{Output}{Output}
\SetAlgorithmName{Algorithm}{alg:detailed_bona}{List of Algorithms}
\caption{Detailed OpGrid Implementation of \bona}
\label{alg:detailed_bona}
\DontPrintSemicolon

\Input{Endpoint DAG $G$ with pending dirty ancillas $Q_D$.}
\Output{Updated DAG $G$ with dirty ancillas mapped to physical wires.}

$\Skipped \gets \emptyset$ \tcp*{Ancillas skipped after a failed global search}

\While{$Q_D \setminus \Skipped \neq \emptyset$}{
    $OpGrid \gets \text{BuildOpMatrix}(G)$ \tcp*{Builds cells, appends dummy
endpoints for $Q \setminus Q_D$}
    $\Pool \gets \text{Empty SortedList}$ \tcp*{Ordered by right endpoints of
intervals}
    $delay\_in \gets \emptyset$ \tcp*{Buffer for transactional pool updates}
    $trouble\_maker \gets \text{NULL}$\;

    \tcp{Sort all cells by (column ASC, is\_dirty\_start DESC, row ASC)}
    $Cells \gets \text{SortedCells}(OpGrid)$\;

    \ForEach{cell $(row, col) \in Cells$}{
        \eIf{cell is the start of a dirty ancilla $d \in Q_D \setminus
\Skipped$}{
            $[fst, lst] \gets \text{lifespan of } d$\;
            $cand \gets \Pool.\text{QueryCoveringInterval}(fst-1, lst+1)$\;

            \eIf{$cand \neq \text{NULL}$}{
                $\text{SoftReuse}(G, OpGrid, cand, d)$ \tcp*{Updates
pointers in $O(1)$}
                $Q_D \gets Q_D \setminus \{d\}$\;
                $delay\_in.\text{add}(\text{right interval of } d)$\;
            }{
                $trouble\_maker \gets d$\;
                \textbf{break} cell loop \tcp*{Soft borrowing failed, halt
sweep}
            }
        }{
            \tcp{Standard operation cell: flush buffer and update intervals}
            $\Pool.\text{add\_all}(delay\_in)$; $delay\_in \gets \emptyset$\;
            $\Pool.\text{remove}(\text{interval ending at } col)$\;
            $\Pool.\text{add}(\text{interval starting at } col)$\;
        }
    }

    \If{$trouble\_maker = \text{NULL}$}{
        \textbf{break} outer loop \tcp*{All feasible ancillas successfully
mapped}
    }

    \tcp{Fallback Phase: Hard Borrowing}
    $d \gets trouble\_maker$\;
    $Fwd \gets \text{Descendants}(G, d.fst) \text{ as } \texttt{set}$
\tcp*{$O(|E|)$ pre-computation}
    $Bwd \gets \text{Ancestors}(G, d.lst) \text{ as } \texttt{set}$\;
    $best\_cand \gets \text{NULL}$\;

    \ForEach{idle interval $(c_l, c_r)$ in $OpGrid$}{
        \If{node at $c_l \notin Fwd$ \textbf{and} node at $c_r \notin Bwd$
\tcp*{$O(1)$ set lookup}}{
            $\estdepthdelta \gets \max(c_l+1-d.fst, 0) +
\max(d.lst+1-c_r, 0)$\;
            Update $best\_cand$ with the interval minimizing
$\estdepthdelta$\;
        }
    }

    \eIf{$best\_cand \neq \text{NULL}$}{
        $\text{HardReuse}(G, OpGrid, best\_cand, d)$ \tcp*{Updates $G$
edges}
        $Q_D \gets Q_D \setminus \{d\}$\;
    }{
        $\Skipped.\text{add}(d)$ \tcp*{Do not repeat the failed global search}
    }
}
\Return $G$\;
\end{algorithm}

\end{document}